\documentclass[11pt, letterpaper, twocolumn]{article}
\usepackage{cite}
\usepackage{graphicx}
\usepackage{lipsum}
\usepackage{multicol}
\usepackage{graphicx}
\usepackage{caption}
\usepackage{refstyle}
\usepackage{amsmath}
\usepackage{amsthm}
\usepackage{amsmath}
\usepackage{amsfonts}
\usepackage{amssymb}
\usepackage{bbm}
\usepackage{amsmath,array}
\usepackage{abstract}
\usepackage{mathtools}
\usepackage{slashbox}
\usepackage[all]{xy}
\usepackage{mathtools}
\usepackage[utf8]{inputenc}
\usepackage[english]{babel}
\usepackage[T3,T1]{fontenc}
\usepackage{fancyhdr}
\usepackage{graphicx} 
\DeclareSymbolFont{tipa}{T3}{cmr}{m}{n}
\DeclareMathAccent{\invbreve}{\mathalpha}{tipa}{16}

\usepackage[left=1in,right=1in,top=1in,bottom=1in]{geometry}
\usepackage{setspace}
\newtheorem{theorem}{Theorem}
\def\mathbi#1{\textbf{\em #1}}

\begin{document}

\title{Quantum-safe Encryption: A New Method to Reduce Complexity and/or Improve Security Level} 
\author{Amir K. Khandani \\ E\&CE Dept., Univ. of Waterloo, Waterloo, Ontario, Canada; khandani@uwaterloo.ca}

\date{}
 \maketitle   

\noindent
\underline{\bf Abstract:}
This article presents some novel techniques to enhance an encryption scheme motivated by classical McEliece cryptosystem. Contributions include: (1) using masking matrices to hide sensitive data, (2) allowing both legitimate parties to incorporate randomness in the public key without sharing any additional public information,  (3) using concatenation of a repetition code for error correction, permitting key recovery with a negligible decoding complexity,  (4) making attacks more difficult by increasing the complexity in verifying a given key candidate has resulted in the actual key,
(5) introducing memory in the error sequence such that: (i) error vector is composed of a random number of erroneous bits, (ii) errors can be all corrected when used in conjunction with concatenation of a repetition code of length 3. Proposed techniques allow generating significantly larger keys, at the same time, with a much lower complexity, as compared to known post-quantum key generation techniques relying on randomization.  

\section{Introduction}
Objective in key sharing is to generate two identical binary vectors at  locations of Alice and Bob, ideally without disclosing any information about the key to Eve.  Typically, keys of size 256 bits are needed to be used with Advanced Encryption System (AES) to encrypt data. Even though Information theory includes a number of  existence results for sharing keys (or secret messages) that are provably unbreakable, but their practical realization remains a challenge.
Due to these challenges, sharing of a symmetric key is typically performed by sending the key as an encrypted message using a public/private key pair; the so-called Public Key Infrastructure (PKI). This operation, refereed to as Key Encapsulation, is the target of attack in extracting the symmetric key. 

Conventional methods for key encapsulation, excluding Quantum Key Distribution (QKD), rely on three sets of techniques: 
 (1) Information theoretical methods that operate based on: (i) adding random noise to data \cite{wire}, or (ii) extracting the information common between two dependent random variables  \cite{common}.   (2) Mathematical methods to construct a one-way function that is hard to invert. (3) Methods motivated by McEliece cryptosystem which incorporate randomness in the public key matrix and rely on error correcting codes. Unfortunately, over time, these complementary sets of techniques have diverged. This paper is an attempt to bridge this gap.

Information theoretical methods, although known to (asymptotically) achieve perfect secrecy, rely on existence results, lacking a clear path to practical realization.  
Turbo-like codes are known
to approach channel capacity bounds with a small loss in energy efficiency. This shortcoming  can be handled by a small increase in signal power. However, unlike the case of channel capacity, it is difficult to quantify and/or deal with any remaining gap in the case of information theoretical security.

Current PKI relies on complex mathematical operations that are (ideally) one way, i.e., it is easy to apply the function, but very difficult to reverse it, unless one has access to a separate piece called the private key. 
In spite of underlying mathematical complexity, there is a one-to-one relationship between the original key and the version that is hidden behind the private key. This one-to-one mapping can be potentially reversed as the computing technology advances. As a remedy, over recent years, a  class of PKI techniques based on including randomness in the public key have received renewed attention. In McEliece Cryptosystem~\cite{R0} (also see \cite{Main-ref} to \cite{RN37} and references therein) and its variants, the randomness is added by randomly rearranging the generator matrix of an error correcting code (public key), such that the known decoding methods cannot be applied unless the randomness in removed using the corresponding private key. Key encapsulation is achieved by encoding a data vector using the publicly known generator matrix, and then an error vector is added to the encoded vector. Error vector falls  within error correcting capability of the underlying forward error correcting code, and hence can be corrected.   Readers are referred to Appendix~\ref{mac} for a review of McEliece Cryptosystem.  

A shortcoming of known McEliece cryptosystems, and its known variants, is that the forward error correcting code is randomized using linear operations, resulting in an equivalent code which can be potentially decoded.
Typically, the randomized code generates the same set of code-words, but with a random permutation of coordinates that renders currently known decoding methods ineffective.  With advancement in decoding of channel codes (see \cite{R3} to \cite{R11}) and continual advancement in computing technology, the danger exists that the underlying code, although modified in appearance, can be decoded. To overcome this shortcoming, it is desirable that the construction of key pair (public and private) incorporate some randomness which can be attacked only through an exhaustive search. This means, in any attempt to recover the actual key from its publicly available version, one has to exhaustively try all random combinations, significantly increasing the overall complexity. One also needs to be able to verify if a key candidate is indeed the actual key (verification phase). The complexity of the verification phase has not yet received the attention it deserves. The current article aims to address this issue. 

This article enhances this procedure is several ways explained in the following: 

Relying on basic McEliece cryptosystem, the proposed method hides the generator matrix of  the error correcting code through masking (addition of a random binary matrix). Such a masking fundamentally changes the generator matrix.  This makes the public key more resistant to attacks, at the same time, complexity of key recovery and storage requirements are significantly reduced as one can rely on simple forward error correcting codes. In particular, a concatenation of repetition codes of length 3 is used. 

In PKI based on classical McEliece cryptosystem and its reported variants, the public key is a randomized by one of the legitimate parties, say Alice. Unlike classical McEliece cryptosystem, in the current article, Bob is also able to introduce randomness  by discarding a randomly selected subset of columns in  the public key generator matrix received from Alice. Bob then uses a random data vector, encodes it using the reduced generator matrix, and adds error to the result. The key point is that, Alice will be able to correct the errors and derive the same key as Bob without knowing which columns are discarded by Bob.   However, Eve needs to exhaustively find the locations of discarded columns before any attempt to break the encryption.

Conventional McEliece cryptosystem relies on addition of an error vector with independent and identically distributed (i.i.d.) components. This is motivated by the model used in forward error correction over memory-less channels. Unlike the case of data transmission, in randomized cryptography, the error vector is constructed by one of the two participating parties, say Bob, and is under his control . This allows introducing memory into the error vector. The introduced memory considered in this article is such that the number of erroneous bits is random, while Alice can perfectly recover all errors. At the same time, the entropy of the error vector is high enough to achieve a highly secure system. Unlike McEliece cryptosystem, which relies on introducing a fixed/known number of errors, the proposed method includes a random number of erroneous bits. Having a random number of erroneous bits adds to the complexity of the attack. 

To improve clarity, in discussions related to classical McEliece cryposystem, italic bold fold notations, e.g.,  $\mathbi{A}$,  are used to represents matrices, 
$\eth$ represents the key data vector and $\epsilon$ represents the added error vector.  In discussions related to the proposed method, regular (capital) boldface notations, e.g.,  $\mathbf{A}$  are used  to represent matrices, and  (lower case) boldface notations are used to represent vectors, e.g., $\mathbf{d}$ represents the key data vector and $\mathbf{e}$ represents the added error vector.

\section{Proposed Method}
To have a basis for comparison, this article focuses on a small part in the overall complexities of McEliece  and Niederreiter Crypto-Systems, i.e., that of matrix multiplication by the public key for key encapsulation. These values are compared with the complexity of  matrix multiplication by the public key for key 
encapsulation in the proposed method. This computation forms the bulk of the complexity in the proposed method, while the bulk of complexity in earlier known techniques is that of recovering the message from the erroneous vector through decoding of the forward error correcting code. This operation has a trivial complexity in the proposed method.   Consequently, presented comparison results are to the disadvantage of the proposed method.

\subsection{Public Key Generation, Key Encapsulation and Recovery}
First, to establish a key between Alice and Bob, it is assumed that Alice is responsible for generating the public key; its associated randomization and then recovering the key from the information received from Bob. On the other hand, Bob 
encapsulates a message using a randomly punctured version of the public key; adds an error vector and sends the result to Alice. Alice then recovers the key. These two operations are explained next. 

\subsubsection{Structure of Public Key}

Alice generates a matrix $\mathbf{A}$ with the structure shown in Fig.~\ref{FigA}.
Alice also generates a matrix $\mathbf{B}$ (see Figs.~\ref{FNB1},\ref{FNB2}), selected to satisfy the conditions in Fig.~\ref{FNAB} for the product $\mathbf{AB}$.   The corresponding public key is equal to ${\mathbf P}={\mathbf B}{\mathbf G}$, with ${\mathbf G}$ to be explained later.  
  \begin{figure}[h]
   \centering
\hspace*{-0.8cm}
   \includegraphics[width=0.4\textwidth]{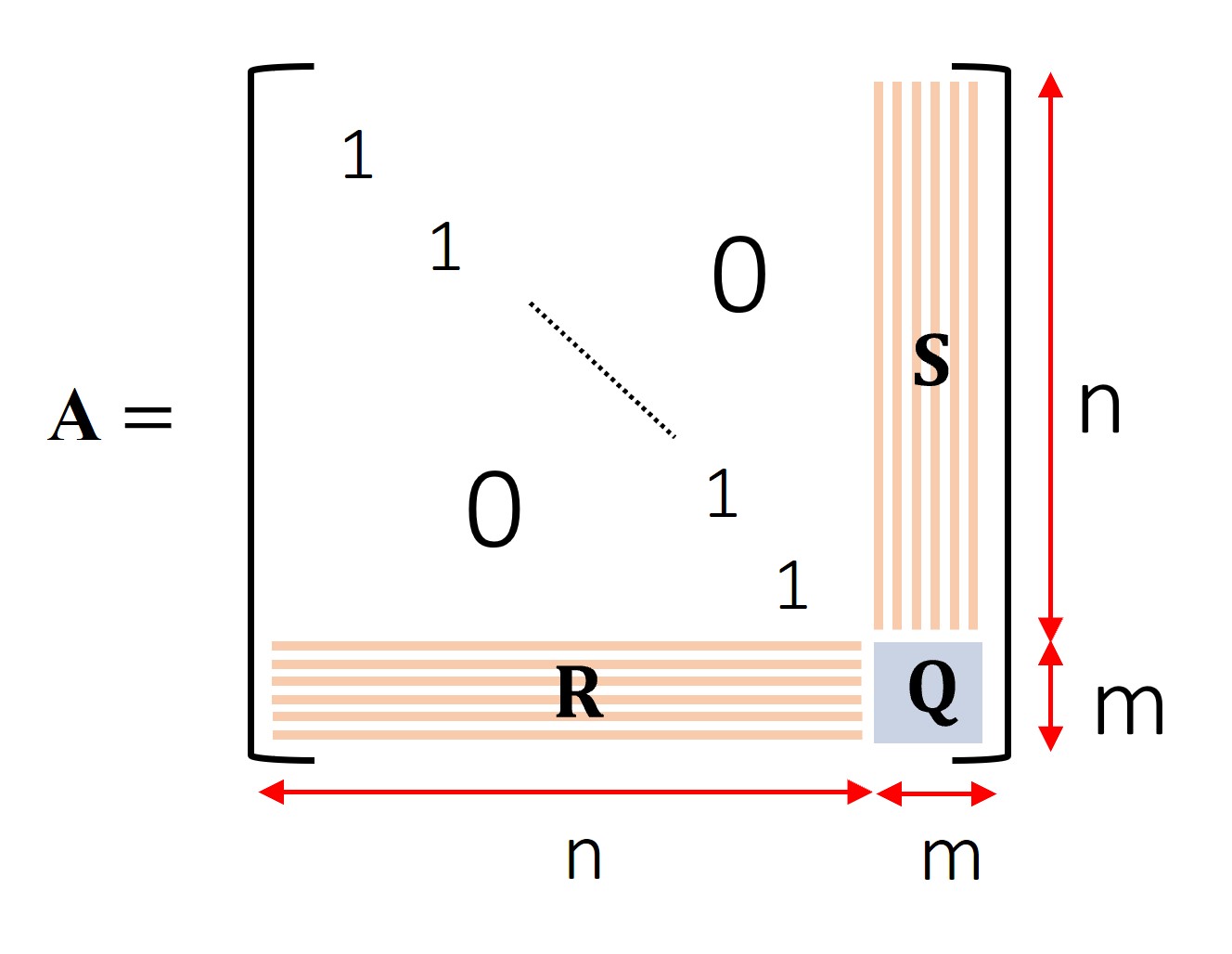}
   \caption{Matrix $\mathbf{A}$ composed of a unitary matrix in the upper left corner, and random rows and columns in $\mathbf{S}$, $\mathbf{R}$ and $\mathbf{Q}$  (used by Alice to recover the key). }
   \label{FigA}
 \end{figure}

\begin{figure}[h]
   \centering
\hspace*{-0.5cm}
   \includegraphics[width=0.4\textwidth]{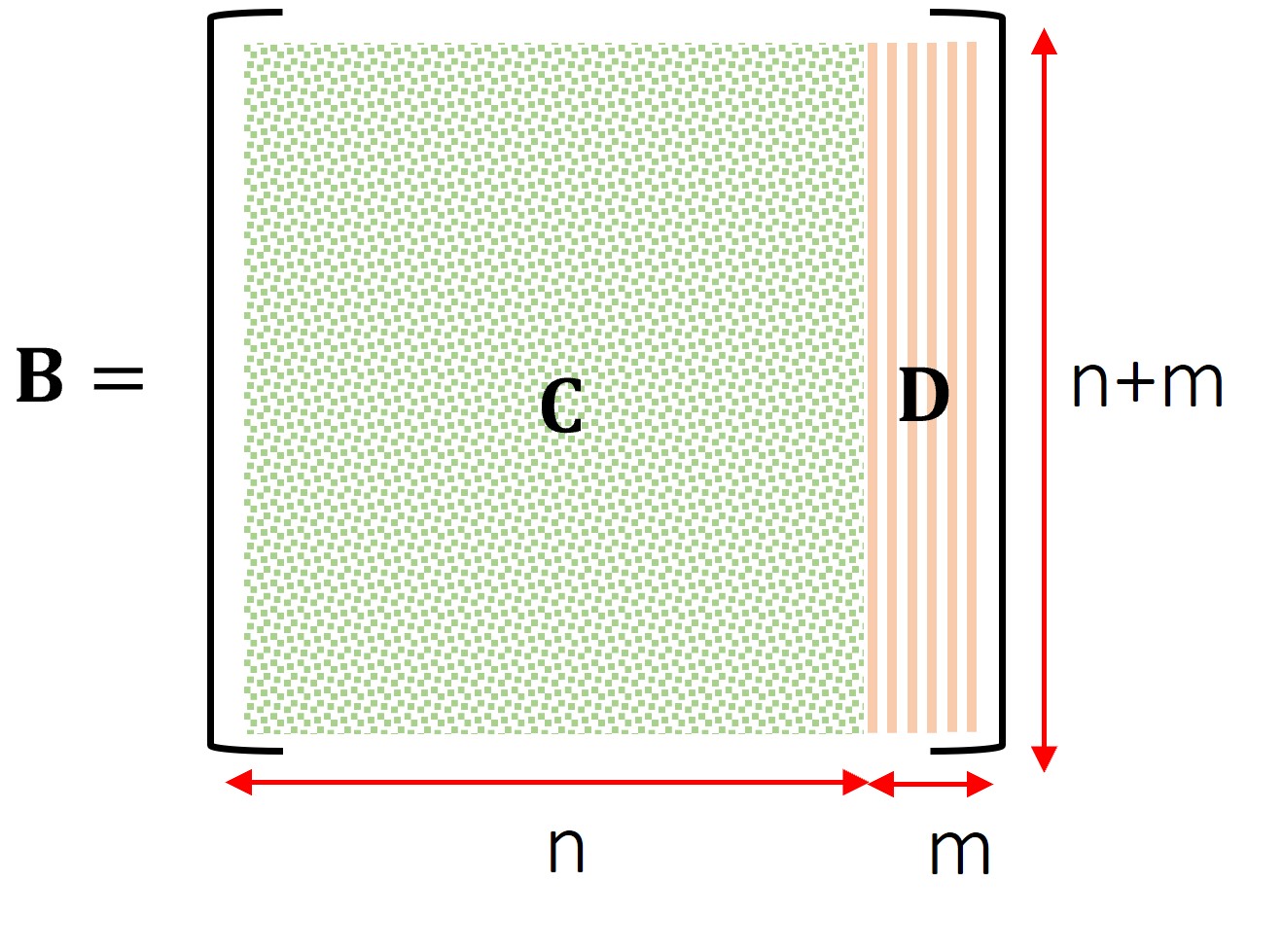}
   \caption{Structure of matrix $\mathbf{B}$ used in generating the public key.}
\label{FNB1}   
 \end{figure} 

    \begin{figure}[h]
   \centering
\hspace*{-0.45cm}
 \includegraphics[width=0.39\textwidth]{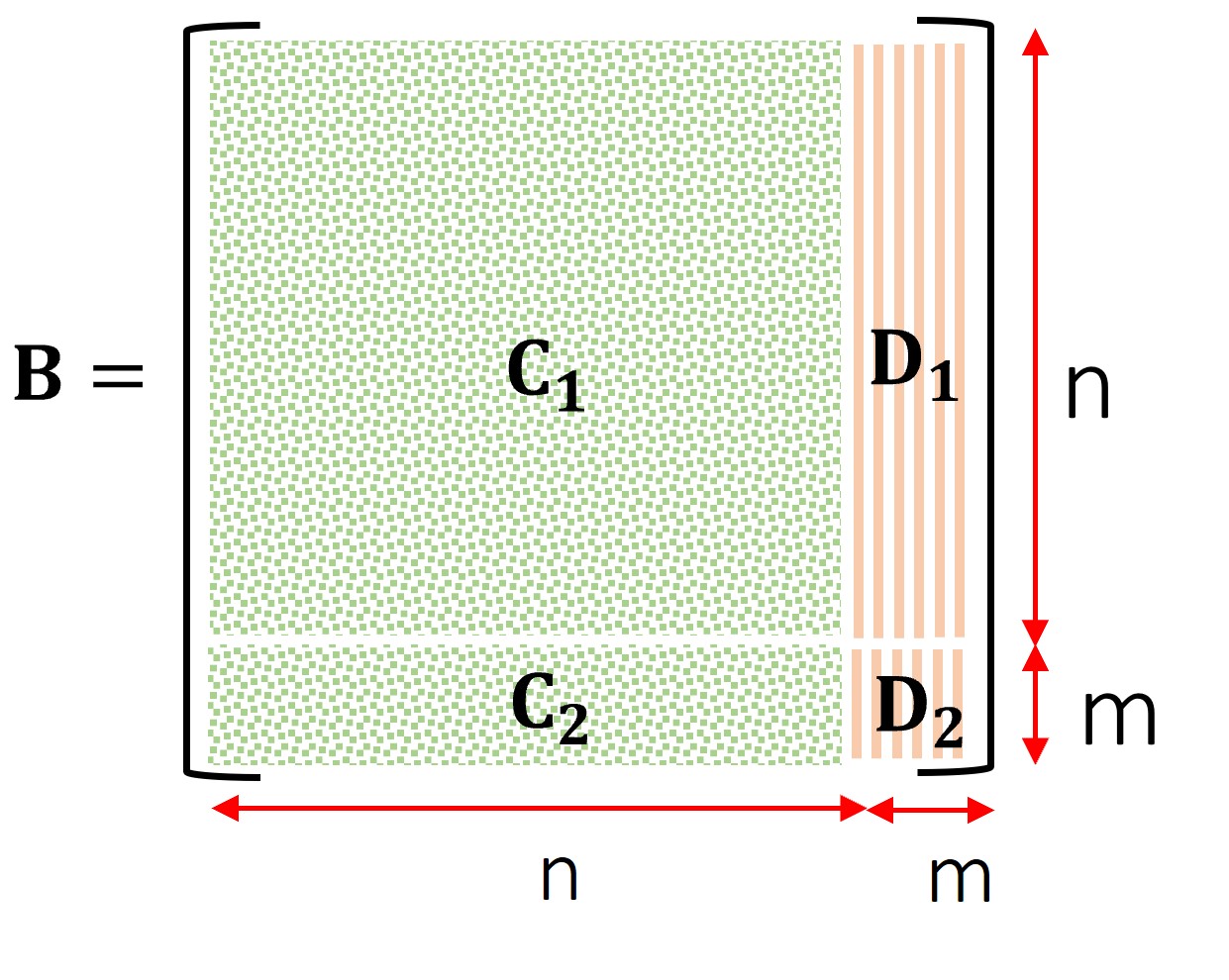}
   \caption{A decomposition of matrix $\mathbf{B}$.}
\label{FNB2} 
 \end{figure} 

    \begin{figure}[h]
   \centering
\hspace*{-0.75cm}
   \includegraphics[width=0.4\textwidth]{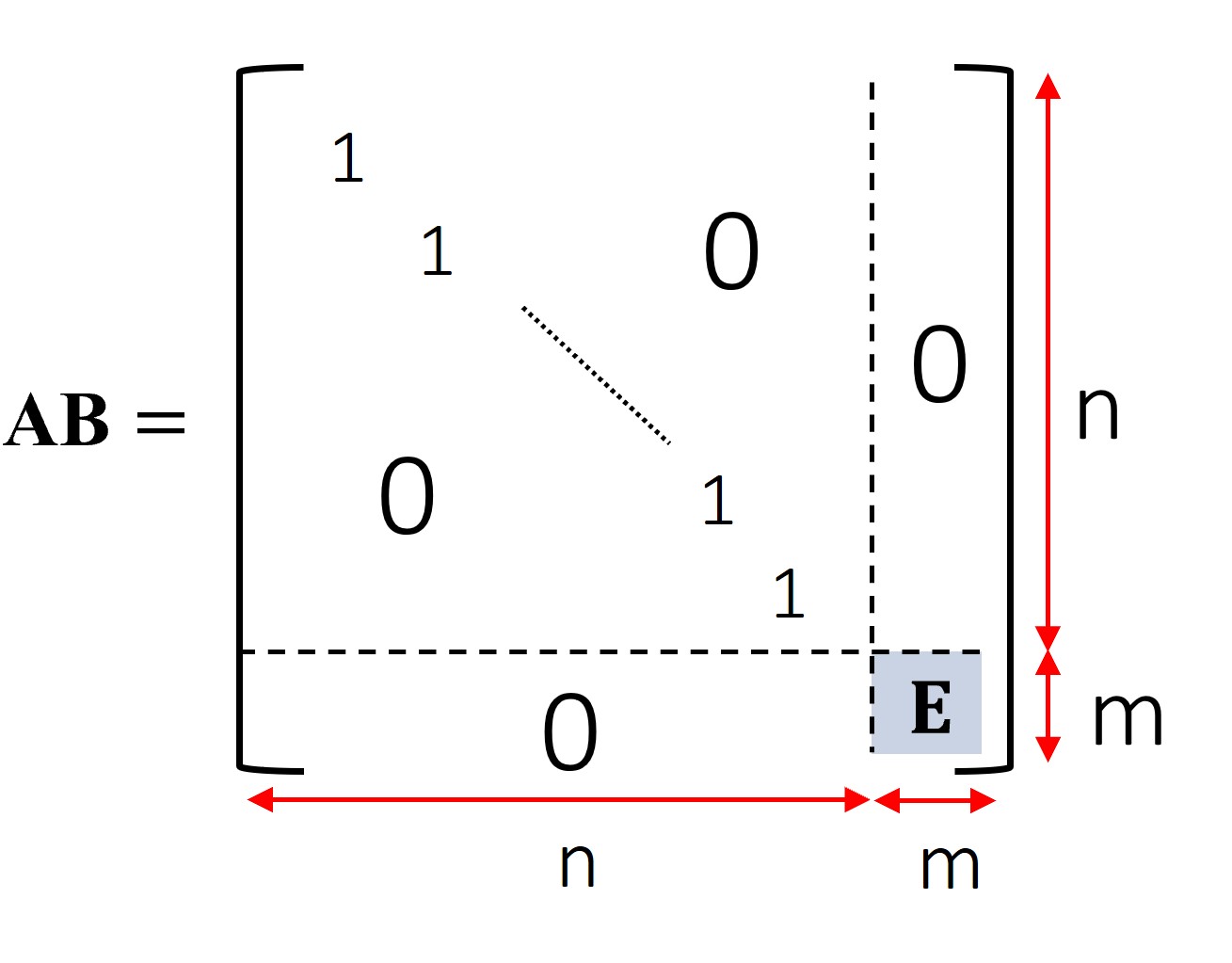}
   \caption{Structure of matrix $\mathbf{A}\mathbf{B}$ formed implicitly, as the 
first step in recovering the key, by Alice (also see Fig.~\ref{FNSet1}).}
\label{FNAB} 
 \end{figure} 

\begin{figure*}[h] 
\centering \hspace*{-0.55cm}\includegraphics[width=0.85\textwidth]{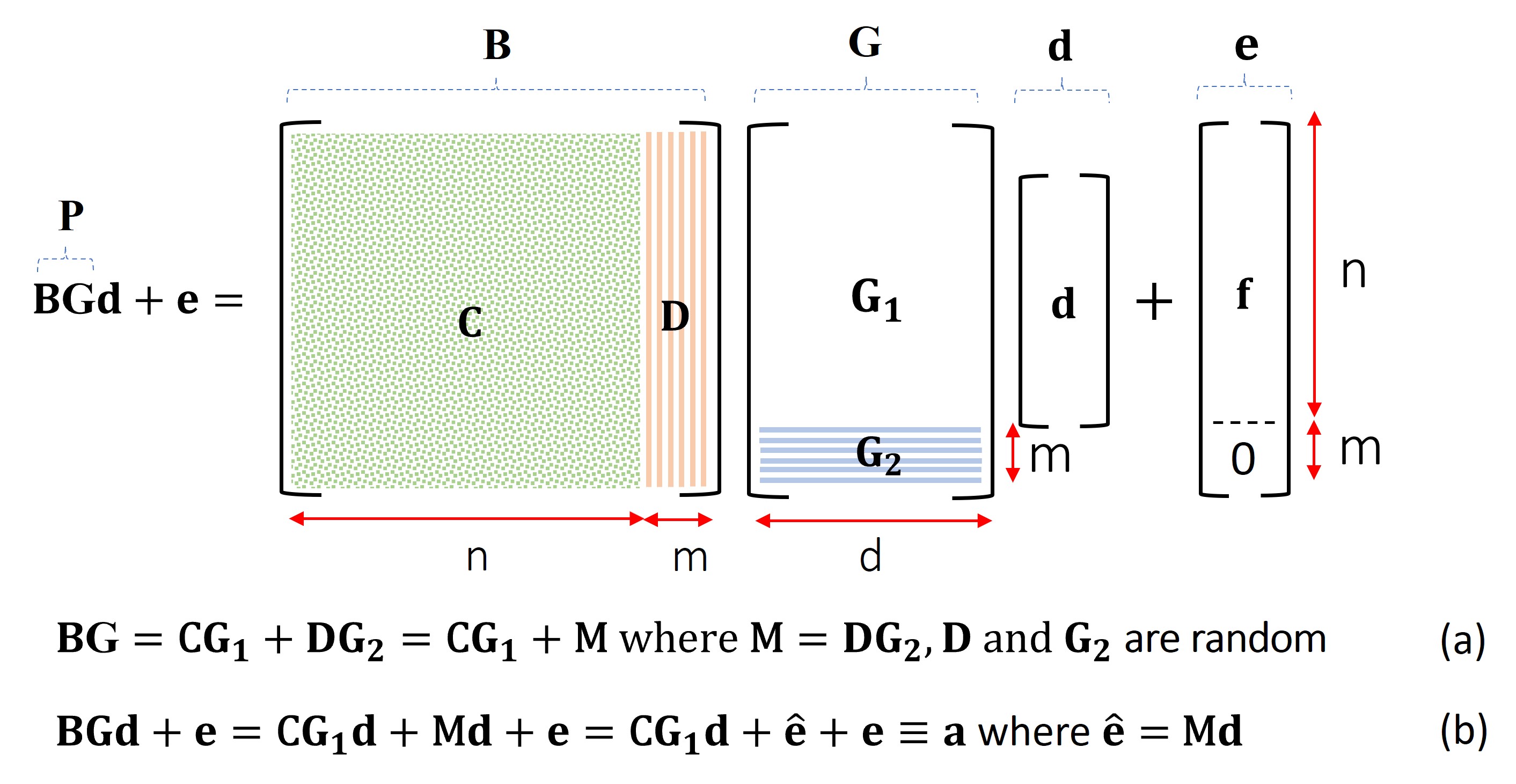} 
\caption{${\mathbf P}={\mathbf B}{\mathbf G}$ is the public key, ${\mathbf e}$ is the error vector following state diagram in Fig. \ref{FigS}, ${\mathbf M}$ is the masking matrix and $\hat{\mathbf e}$ is the error vector capturing the effect of virtual errors due to ${\mathbf M}$. To add further confusion, Bob removes $\mathsf{p}$  randomly selected columns from ${\mathbf P}$ to generate $\hat{\mathbf P}$ which is then multiplied by a shortened data vector $\hat{\mathbf d}$ of length $\mathsf{d-p}$. }
\label{FNSet2} 
\end{figure*}

Error correcting code used here is constructed by concatenation of $\mathsf{d}$ repetition codes of lengths 3. Each column of the generator $\mathbf{G}_1$  (see Fig. \ref{FNSet2}) is composed of a single repetition code of length 3, and bits constructing repetition codes in different columns do not overlap. The resulting row size of matrix $\mathbf{G}_1$ is equal to $\mathsf{n=3d}$,  $\mathbf{G}_1$ is an 
$\mathsf{n}\times \mathsf{d}$ matrix.  
Then columns of $\mathbf{G}_1$ are randomly permuted by Alice. 
Alice then adds $\mathsf{m}$ random rows at the end of the resulting $\mathbf{G}_1$ to construct the composite generator matrix $\mathbf{G}$ of size $\mathsf{n+m}\times \mathsf{d}$ (see Fig.~\ref{FNSet2}). Figures~\ref{FNAB}, \ref{FNSet1} include  a matrix ${\mathbf E}$ to be discussed/used later.  Finally, structure of the public key, i.e.,  $\mathbf{BG}$  is depicted in Fig.~\ref{FNSet2}.

\subsubsection{Key Encapsulation by Bob}

Having access to the public key, i.e., ${\mathbf P}={\mathbf B}{\mathbf G}$, 
Bob first selects $\mathsf{p}$ random  indices among columns of ${\mathbf P}$ and discards those columns. This random selection changes in each round of  key encapsulation. This allows for updating the public key by Bob, without the need for sending any pubic data. Random discarding of public key columns plays an important role in increasing the entropy in key encapsulation; as well as in increasing the attack difficulty. For simplicity, whenever removing operation does not play a direct role, the article relies on notations related to the original key and its associated components e.g., in Figs.~\ref{FNSet1},\ref{FNSet2}.  

\begin{figure}[h]
   \centering
\hspace*{-.2cm}
   \includegraphics[width=0.32\textwidth]{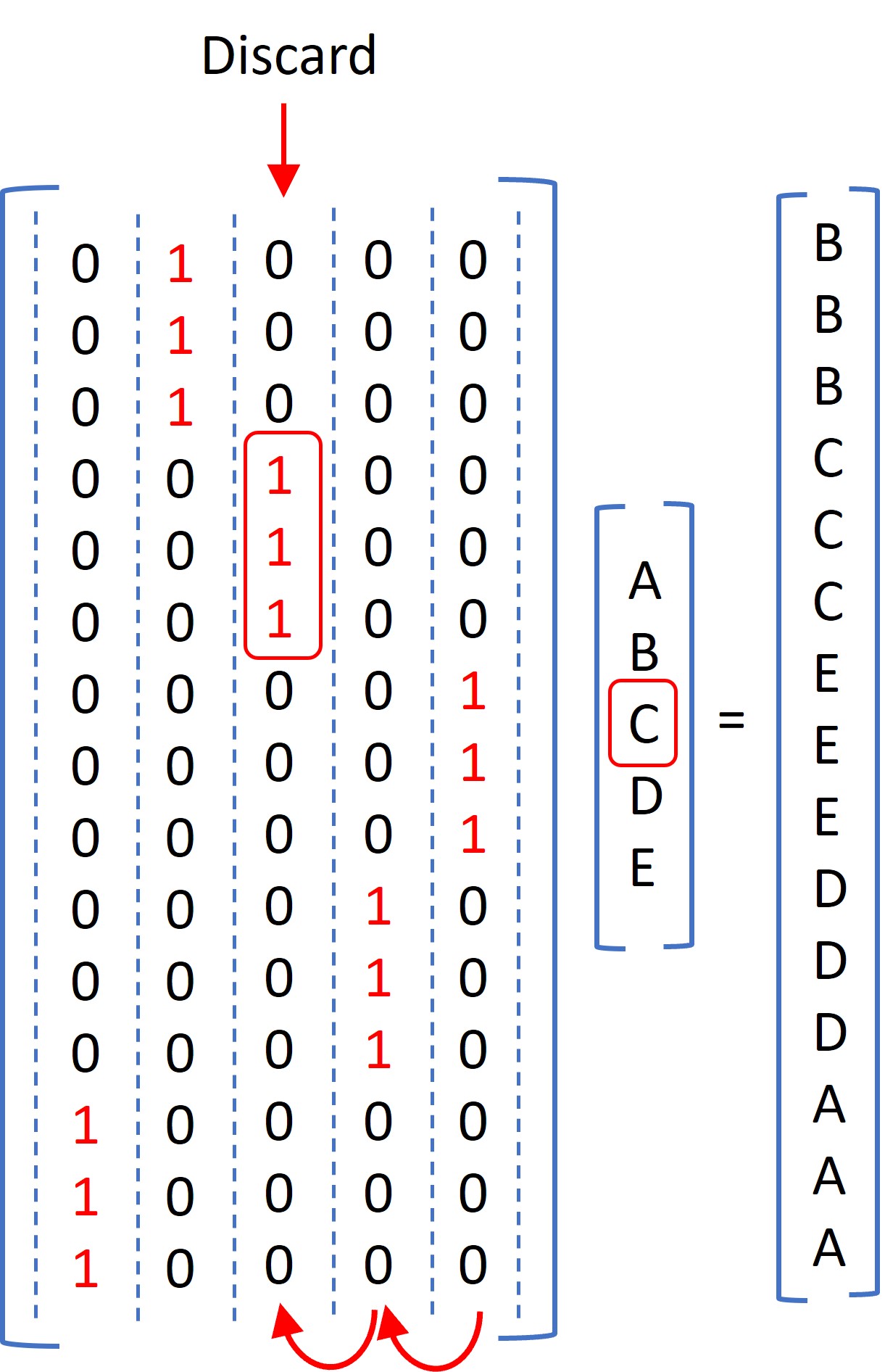}
   \caption{An example for a generator matrix representing column-wise permuted concatenation of length 3 repetition codes .}
   \label{Pru1}
 \end{figure}

 \begin{figure}[h]
   \centering
\hspace*{-.2cm}
   \includegraphics[width=0.29\textwidth]{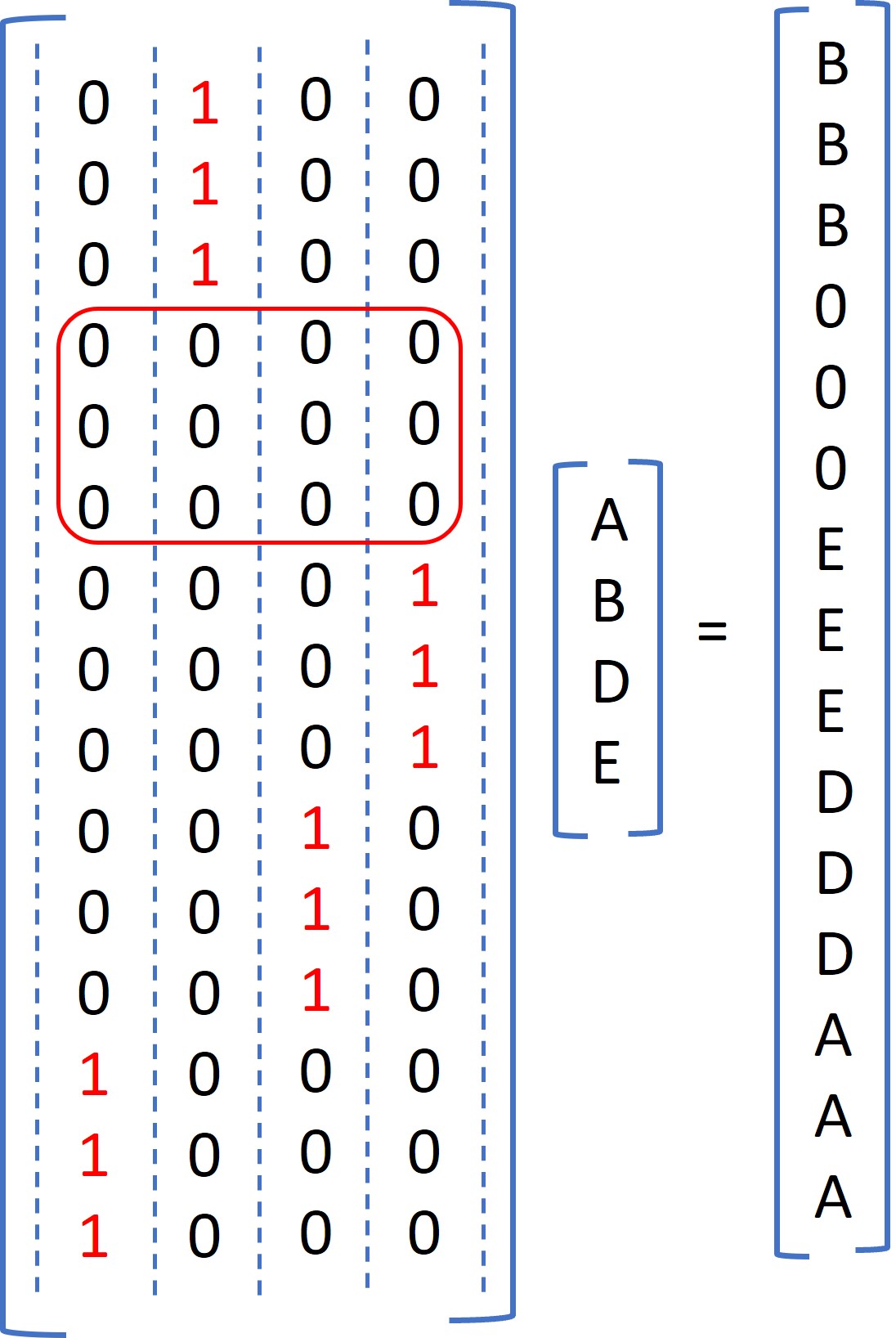}
   \caption{Effect of discarding a column in a concatenation of length 3 repetition codes in Fig. \ref{Pru1}.}
   \label{Pru2}
 \end{figure}

The reduced public key is denoted as $\hat{\mathbf P}$. Removing ${\mathsf p}$ columns from ${\mathbf P}$ is accompanied by using a shortened vector,  
$\hat{\mathbf d}$, of size 
$\mathsf{d}-\mathsf{p}$  to be multiplied by the reduced generator matrix.   
Bob encapsulates a randomly selected vector $\hat{\mathbf d}$ by computing 
$\hat{\mathbf P}\hat{\mathbf d}$, generates an error vector ${\mathbf e}$ following the state diagram in Fig.~\ref{FigS}, adds $\mathsf{m}$ zeros at the end of vector ${\mathbf e}$ (see Fig.~\ref{FNError}(a)). Then, Bob adds the result to the vector $\hat{\mathbf P}\hat{\mathbf d}$ and sends 
$\hat{\mathbf P}\hat{\mathbf d}+{\mathbf e}$ to Alice. 
Note that removing columns does not change the size of vectors 
$\hat{\mathbf P}\hat{\mathbf d}$ and ${\mathbf e}$, i.e., the size will be equal to 
$\mathsf{n}+\mathsf{m}=\mathsf{3d}+\mathsf{m}$.

Theorem \ref{ThNAK1} in Appendix~\ref{Theo1} establishes that in a chain of  multiplied matrices, say 
$\hat{\mathbb{M}}=\mathbb{M}_\alpha...\mathbb{M}_2\mathbb{M}_1$, removing a column indexed by  $\beta$ in $\mathbb{M}_1$ results in removing column $\beta$ in 
$\hat{\mathbb{M}}$. In case multiple columns are removed, a similar effect occurs for each removed column. The impact on the result of multiplication is the same as having a zero in positions within the original data vector ${\mathbf d}$ that correspond to indices of removed columns prior to computing ${\mathbf P}{\mathbf d}$. 

Figures~\ref{Pru1}, \ref{Pru2} show the effect of removing a column, and its net effect at the Alice's side upon multiplication by $\mathbf{A}$. As mentioned, Fig.~\ref{Pru2} shows that the removed column (third element, corresponding to the third column in Fig.~\ref{Pru1}) results in filling the corresponding three positions within the result of multiplication in Fig.~\ref{Pru2} with zeros. 
Referring to Figs.~\ref{Pru1},\ref{Pru2}, to arrive at a key that is the same as the version constructed at the Alice's side, Bob privately inserts $\mathsf{p}$ zeros within $\hat{\mathbf d}$ at positions corresponding to removed columns, increasing the key size to $\mathsf{d}$.  This operation is required to guarantee that Alice and Bob can derive the same key, although the positions of discarded columns remain a secret to public, as well as to Alice. Since the reconstructed erroneous vector at the Alice side is composed of a concatenation of repetition codes added to an error vector constructed using the state diagram in Fig.~\ref{FigS}, Alice will be able to correct all errors (since that each error bit, i.e., 1 is followed by at least two zeros Fig.~\ref{FigS}).

 As mentioned, the effect of removing of columns is equivalent to inserting three consecutive zeros at the corresponding positions within the vector $\hat{\mathbf P}\hat{\mathbf d}$. As a result, any repetitions of three zeros could be caused by an actual zero being part of vector $\hat{\mathbf{d}}$ formed by Bob, or due to removing of the corresponding column. Alice does not know which of these two is the actual case, but it does not matter since Alice will be able to correct all errors, resulting in a vector composed of zeros in all bits of  $\hat{\mathbf{d}}$ that were set to zero by Bob, as well as the zeros due to removed columns. Since Bob has (privately) inserted zeros in positions corresponding to discarded columns, upon removing the last $\mathsf{m}$ bits, Alice and Bob arrive at the same key of length $\mathsf{n}$. Alice and Bob have the option of using a hash function to reduce the size of the key to a shorter key aimed at having equal probabilities for zero and one, and an entropy capturing the entropy of the original key.

 \begin{figure}[h]
   \centering
\hspace*{-.2cm}
   \includegraphics[width=0.39\textwidth]{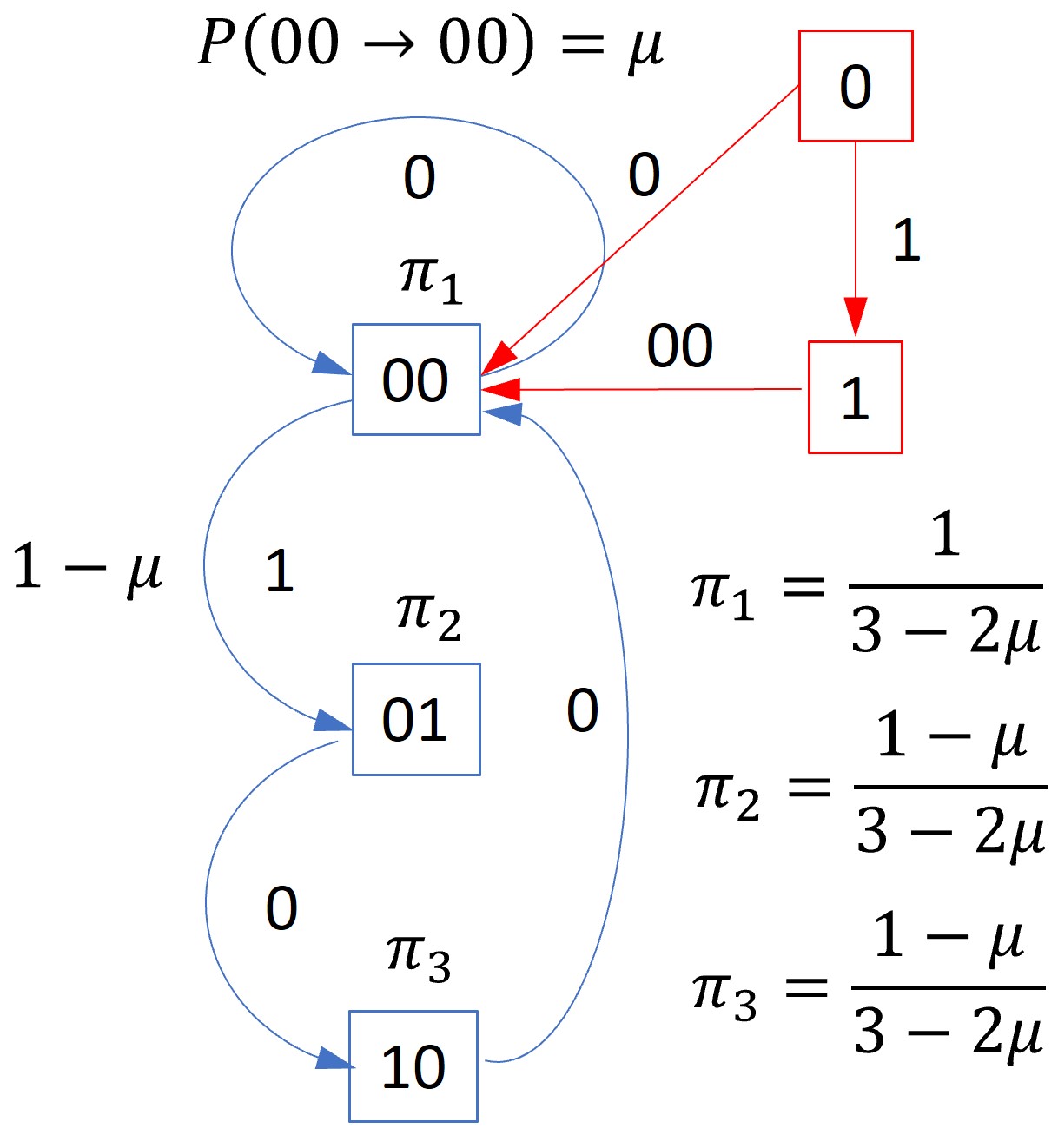}
   \caption{State diagram for generating $\mathbf{e}$ . Note that each error bit, i.e., a one, is followed by at least two zeros. Boxes in red depict the formation of the starting bits in 
$\mathsf{e}$ and can be ignored.}
   \label{FigS}
 \end{figure}

    \begin{figure}[h]
   \centering
\hspace*{-0.0cm}
   \includegraphics[width=0.38\textwidth]{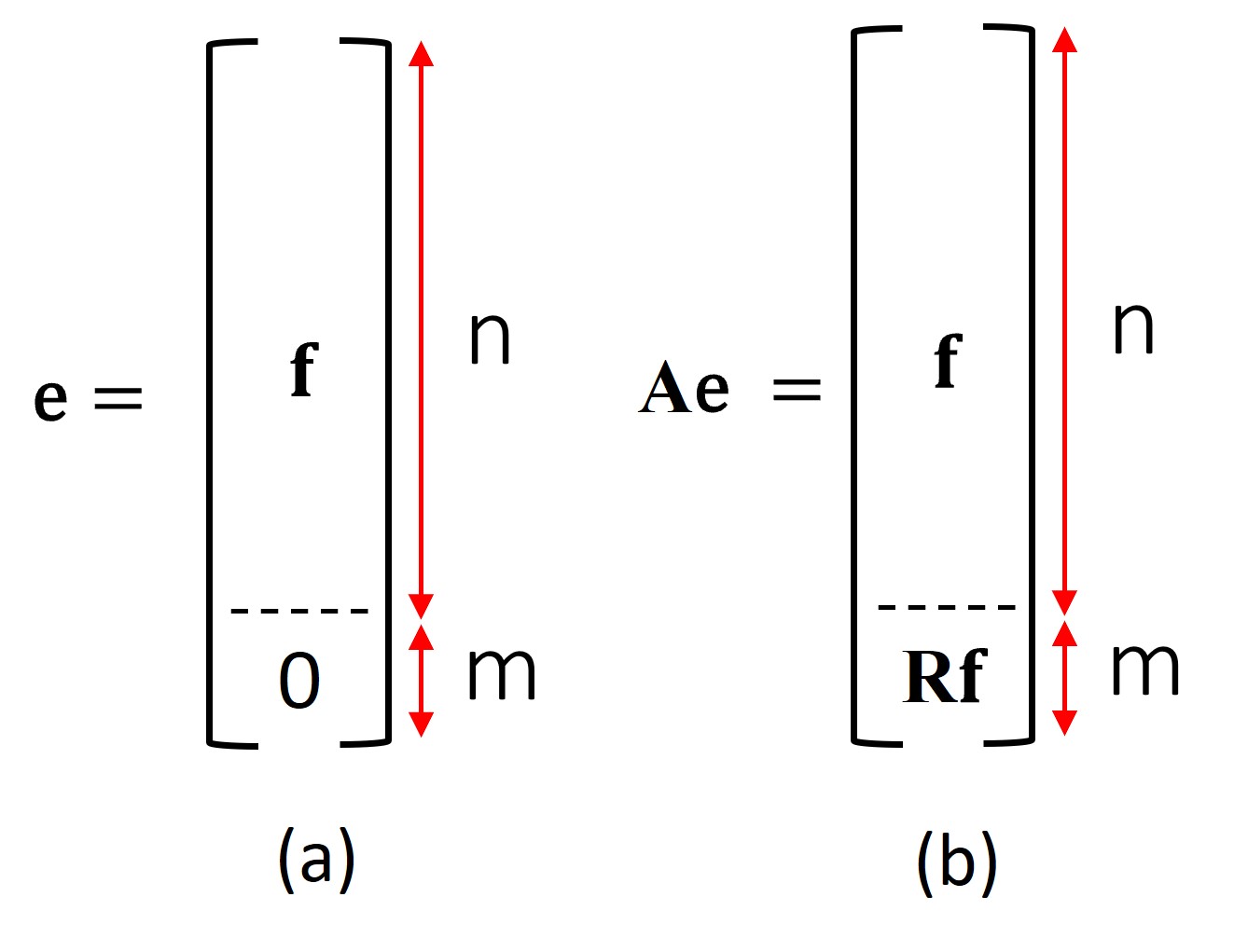}
   \caption{(a) Structure of the error vector $\mathbf{e}$ added to the public key, where  $\mathbf{f}$ is a binary vector with memory constructed based on the state diagram in Fig.~\ref{FigS}. (b) Modified error vector upon multiplication of public key with matrix $\mathbf{A}$ at the Alice side.}
\label{FNError} 
 \end{figure} 

    \begin{figure}[h]
   \centering
\hspace*{-0.0cm}
   \includegraphics[width=0.4\textwidth]{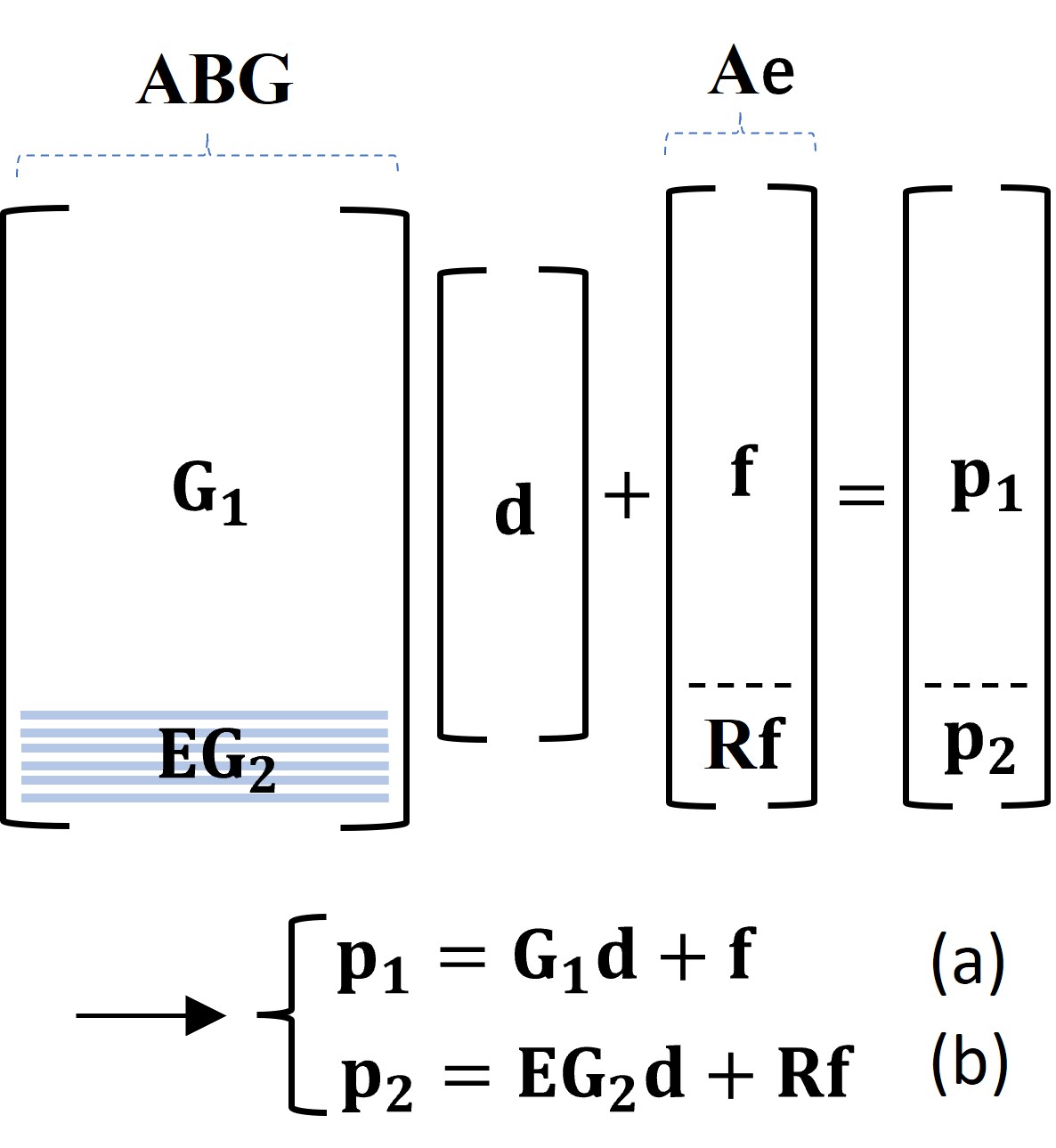}
   \caption{Operations used in recovering the encrypted message. Note that the product of matrices ${\mathbf A}\mathbf{B}$ is as shown in Fig.~\ref{FNAB}. 
See Figs.~\ref{FNAB},\ref{FNSet1} for the definition of  matrix ${\mathbf E}$.} 
\label{FNrecover} 
 \end{figure} 

\begin{figure*}[h] 
\centering \hspace*{0cm}\includegraphics[width=1\textwidth]{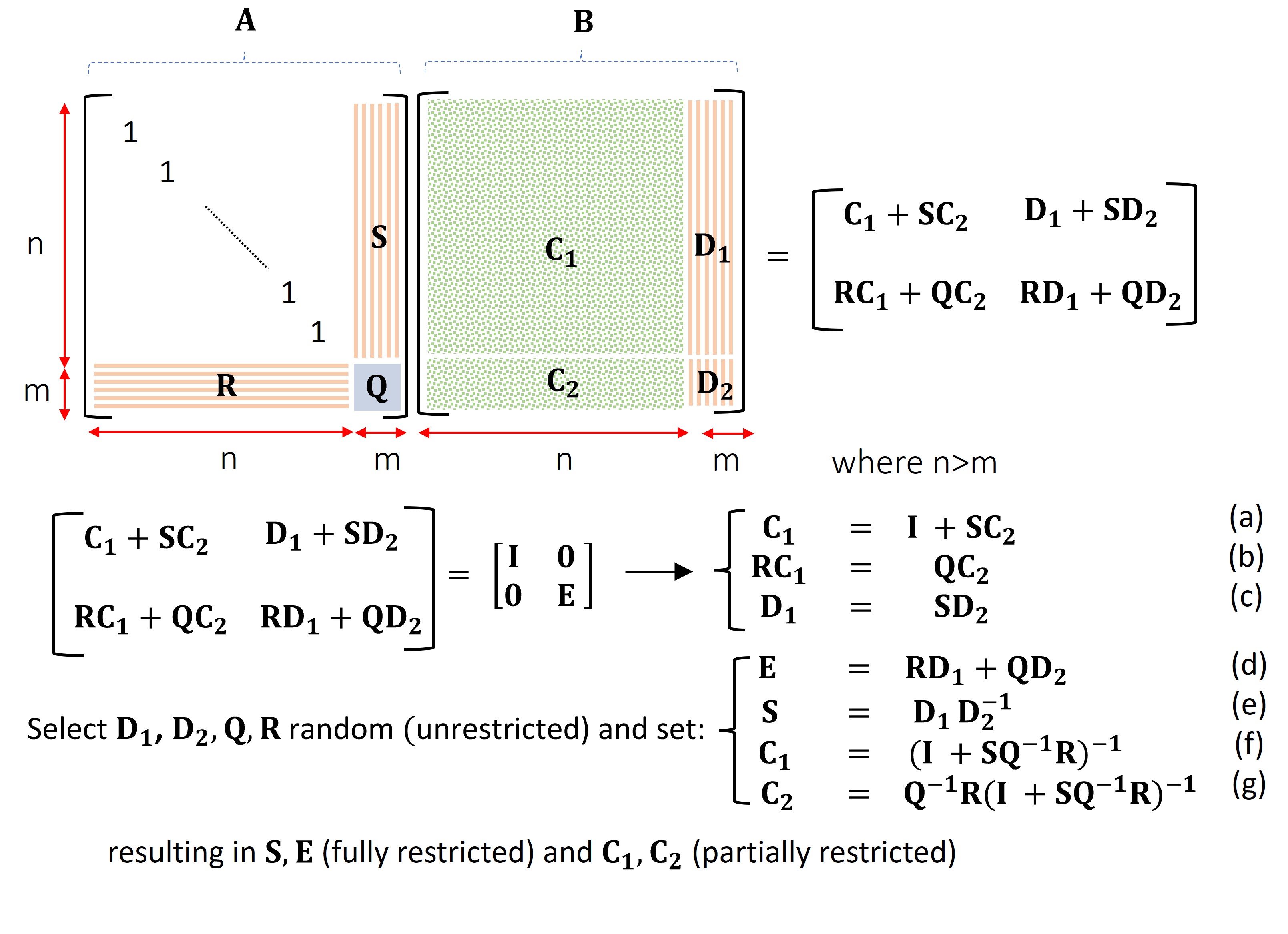} 
\caption{Conditions governing matrices ${\mathbf A}$ and ${\mathbf B}$ resulting in recovery expressions  in Fig.~\ref{FNrecover}.} 
\label{FNSet1}
\end{figure*}

\subsubsection{Key Recovery by Alice}

Alice starts the key recovery process by multiplying the received vector (constructed by Bob using the shortened vector, $\hat{\mathbf d}$) by matrix  ${\mathbf A}$.  Note that the effect of random permutation applied to the columns of ${\mathbf G}_1$ does not affect the outcome, i.e., repetition code occupying rows
$\kappa$, $\kappa+1$,  $\kappa+2$ in column $\xi$ act as a repetition code with bits ordered consecutively, multiplied by the bit at position $\xi$ of vector ${\mathbf d}$ and the result occupies positions $\kappa$, $\kappa+1$,  $\kappa+2$ in the vector obtained as a result of this multiplication. 

According to Theorem \ref{ThNAK1} in Appendix~\ref{Theo1}, the effect of removing the column  $\varkappa$ in the public  key matrix is equivalent to removing column 
$\varkappa$ in  ${\mathbf G}_1$. It results in 3 consecutive zeros in positions $\varkappa$,$\varkappa+1$ and $\varkappa+2$ in the outcome. Since the resulting structure resembles a concatenation of repetition codes of length 3,  Alice will be able to correct all errors. Note that the consecutive zeros in positions 
$\varkappa$,$\varkappa+1$ and $\varkappa+2$ of the resulting vector is equivalent to having a zero at the position $\varkappa$ in the original vector ${\mathbf d}$. 
Since two consecutive ones cannot appear in the error vector (see state diagram in Fig.~\ref{FigS}), one can conclude that addition of the error vector never turn a repetition of three ones (which can occur when the corresponding bit in data stream, ${\mathbf d}$, is a one)  into three zeros. Consequently, three consecutive zeros in the vector obtained after error correction could occur in one of two cases: (i) Corresponding bit in the extended ${\mathbf d}$ is zero, or (ii) corresponding column is among  columns removed by Bob. In either of these cases, Alice decides for a zero in the corresponding position within extended ${\mathbf d}$. Bob does the same, i.e., inserts zeros in the shortened encapsulated key at positions corresponding to discarded columns. As a result, Alice and Bob gain access to the same bit stream of size $\mathsf{n}=3\mathsf{d}$.  However, the resulting vectors have a higher probability for zero vs. one, due to including extra zeros in positions corresponding to columns removed by Bob.  Alice and Bob can use the resulting extended keys of size $\mathsf{d}$, or rely on some form of hashing to reduce the size to the actual entropy level of $\mathsf{d-p}$.  Readers are refereed to Section~\ref{compAK} for comparisons between proposed method with McEliece  and Niederreiter Crypto-Systems. Table~\ref{TabN1AK} includes some examples including key entropy and $\mathsf{SEC}$ ($\mathsf{SEC}$ is defined in expression \ref{AE4AK00}) value.

\subsubsection{Properties of Error Vector $\mathbf{f}$}

Role of $\mathbf{f}$ is to add a binary error vector with a particular memory which can be corrected by the underlying repetition code.  Noting the procedure involved in key recovery in Fig.~\ref{FNrecover}, at the Alice's side, the upper part of the error vector $\mathbf{f}$ will be multiplied by an identity matrix, and its lower part by matrix ${\mathbf R}$ which is the random matrix at the lower part of ${\mathbf A}$ (see Fig.~\ref{FigA}). This results in adding error to the last $\mathsf{m}$ bits. These bits will be discarded by Alice. As a result, the structure of the error vector $\mathbf{e}$ remains unchanged. The overall error vector is depicted in Figs.~\ref{FNError},\ref{FNrecover}. Noting the state diagram in Fig.~\ref{FigS}, each error bit (a bit of value one) is followed by at least two zeros.  Consequently, it is easy to see that a repetition code of length 3  can correct  all the erroneous bits. 
Since the last $\mathsf{m}$ bits are discarded by Alice, terms ${\mathbf R}{\mathbf f}$ and $\mathbf{P}_2$ are not given any consideration in this article.

\subsubsection{Benefits of Using $\mathbf{A}$ in the Construction of Public Key}

The aforementioned formation of public key using matrix $\mathbf{A}$ offers two benefits: 

\noindent {(1)} It allows keeping the first $\mathsf{n}$ components of the added error vector unaffected in the recovery operation performed by Alice. On the other hand, the operation used in key recovery depicted in Fig.~\ref{FNSet1}, brings back the structure of the generator matrix $\mathbf{G}_1$ to its original form. Alice can correct the added error vector   $\mathsf{f}$, and then, knowing how the columns of $\mathsf{G}_1$ are permuted, Alice can recover the encapsulated key selected by Bob. 
\newline
{(2)}  It spreads the randomness inserted in $\mathbf{R}$, $\mathbf{S}$ and 
$\mathbf{Q}$ (see Fig. \ref{FigA}) throughout the public key matrix  $\mathbf{BG}$. Indeed, in the sense discussed in Theorem \ref{fullrank2}, randomness will be uniformly distributed within the constructed masking matrix.
  
\section{Error Vectors with Memory}
Coding theory has been developed based on memory-less channels, i.e., error vector is composed of unknown i.i.d. bits to be detected at the receiver. Use of memory-less error vectors has propagated to the application of  forward error correcting codes in code-based cryptography. However, there is a difference in these two domains. In data transmission over a memory-less channel, the error sequence is out of transmitter's and/or receiver's control. However, in code-based cryptography, error vector is constructed by one of the legitimate parties and and thereby can be controlled.  In the following, this feature is exploited as a tool to improve code-based cryptography.  

\subsection{An Error Sequence Correctable by a Concatenation of Repetition Codes of Length 3}

Using the state diagram  in Fig.~\ref{FigS}, Bob generates an error vector that is completely detectable by a concatenation of repetition codes of length 3. The  constraint imposed by this state diagram is that an erroneous bit (one in the error sequence) is always followed by at least two error-free bits (two consecutive zeros in the error sequence). Note that erroneous bits can still occur in any position within the vector $\mathbf{e}$. It is easy to see that such an error sequence can be always corrected when a number of repetition codes of length 3 are concatenated. 

As mentioned earlier, properties of this state diagram, captured in the value of $\mu\in[0,1]$, govern two factors: (1) entropy of the error vector (determines the complexity in an exhaustive search attack), and (2) number of ones in the error vector (determines the complexity in an information set decoding attack). Figs.~\ref{Cross260} to \ref{Cross1000} show how these two criteria vary by changing $\mu$. The curves cross at a point that determines the security level, i.e.,  $\mathsf{SEC}$ ($\mathsf{SEC}$ is defined in expression \ref{AE4AK00}). Note that the entropy associated with identifying the location of discarded columns is included as an additive factor of 
\begin{equation}\label{additive} 
\mathsf{E}=\log_2{\mathsf{n} \choose \mathsf{p}}
\end{equation}
in both curves in Figs.~\ref{Cross260} to \ref{Cross1000}. 

 
\begin{figure}[h]
   \centering
\hspace*{-.35cm}
   \includegraphics[width=0.5\textwidth]{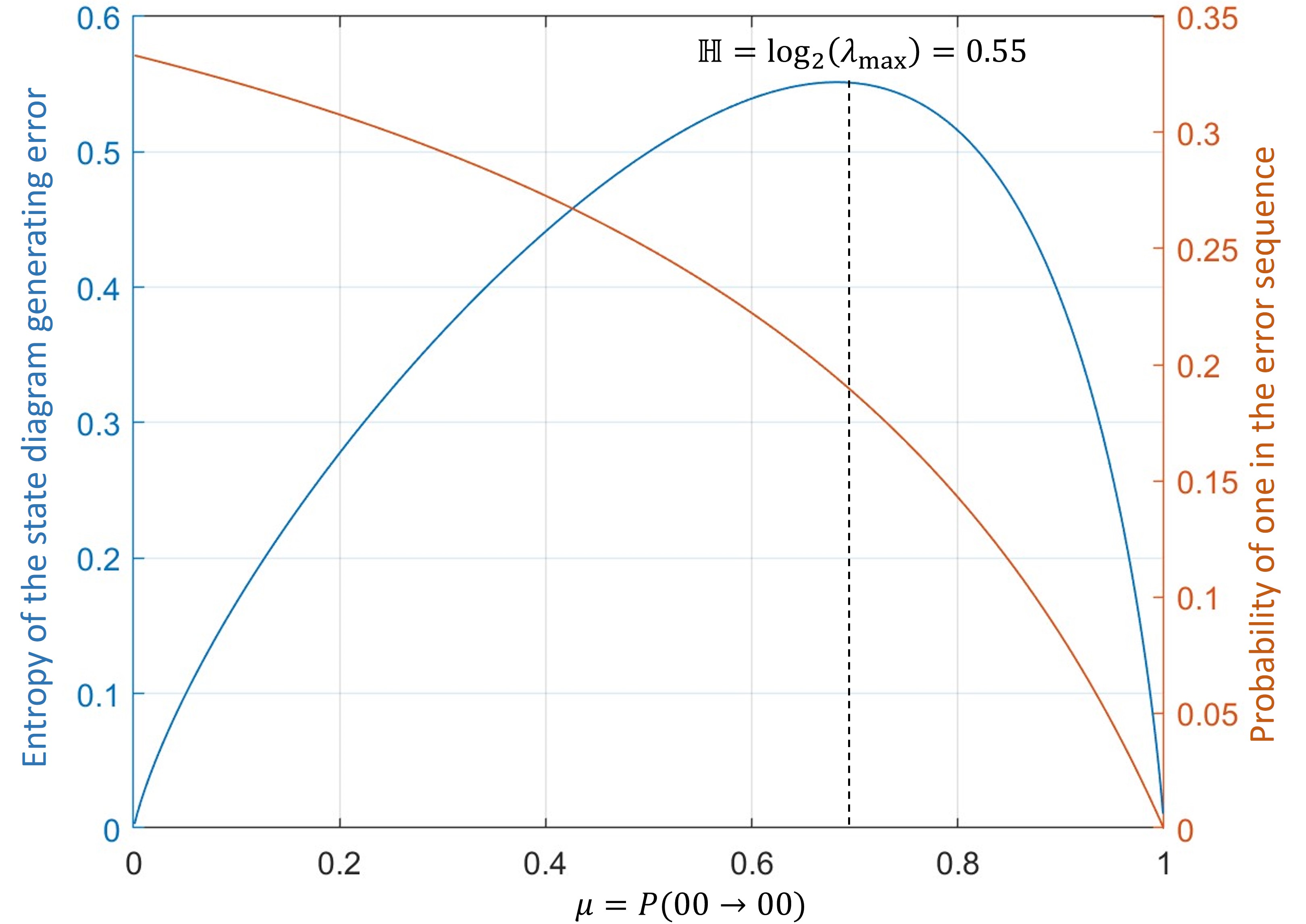}
   \caption{Entropy of the binary sequence (left Y-axis) and probability of one (right Y-axis) generated by the state diagram in Fig.~\ref{FigS}.}
   \label{FigPEE}
 \end{figure}

\subsubsection{Entropy of Error Sequence} \label{SECAK1}
In Fig.~\ref{FigS}, parameter $\mu$ determines  $\pi_1$, $\pi_2$ and $\pi_3$, i.e., probabilities of being in states 00, 01 and 10, respectively. 
Following standard methods of information theory~\cite{cover}, we have
\begin{eqnarray}
\pi_1 & = & \pi_1\mu+\pi_3 \\
\pi_2& = &\pi_1(1-\mu) \\
\pi_3& = &\pi_2 \\
\pi_1+\pi_2+\pi_3 & = & 1
\end{eqnarray}
\begin{eqnarray} \label{Ent1AK}
\pi_1 & = & \frac{1}{3-2\mu} \\  \label{Ent2AK}
\pi_2& = &\frac{1-\mu}{3-2\mu} \\  \label{Ent3AK}
\pi_3& = &\frac{1-\mu}{3-2\mu}.
\end{eqnarray}
Noting  that 
\begin{equation}
\mathrm{H} (\pi_2)=\mathrm{H} (\pi_3)=0.
\end{equation}
Entropy of the state diagram is equal to:  
\begin{equation}
\mathbb{H}_{s}(\mu)= \displaystyle  
\sum_{i=1}^{3} \pi_i 
\mathrm{H} (\pi_i)  = \pi_1 \mathcal{H}(\mu)=\frac{\mathcal{H}(\mu)}{3-2\mu}
\label{Ent3}
\end{equation}
where  $\mathcal{H} (\mu)$ is the binary entropy function defined as 
\begin{equation}
\mathcal{H} (\mu) =-\mu \log_2(\mu)-(1-\mu) \log_2(1-\mu).
\label{AK2}
\end{equation}
Expression \ref{AK2} is the information bit per binary symbol generated by the state diagram.  Result of the entropy vs. $\mu$ is plotted in Fig.~\ref{FigPEE} (left Y-axis).
Using standard techniques,  the probability of 1 in the error sequence, as a function of $\mu$, is computed as
\begin{equation}
\pi_1(1-\mu)=\frac{1-\mu}{3-2\mu},
\label{AK2AK}
\end{equation}
which is also plotted in in Fig.~\ref{FigPEE} (right  Y-axis).

The adjacency matrix of the state diagram in \ref{FigS} is
\begin{equation}
\left[
    \begin{array}{ccc} 
    1 & 1 & 0 \\
    0 & 0 & 1 \\
1 & 0 & 0 \\ 
\end{array} 
\right].
\label{mat-eig}
\end{equation}
Using standard methods of information theory, the corresponding maximum entropy per bit is computed as~\cite{cover}
\begin{equation}
\log_2(\lambda_{\max})=0.55
\label{maxE}
\end{equation}
where $\lambda_{\max}$ is the maximum eigenvalue of the matrix in \ref{maxE} (see the peak in the entropy curve  in Fig.~\ref{FigPEE}).


McEliece, in his groundbreaking work~\cite{R0}, introduces the information set decoding attack which relies on finding a set of equations, formed by selecting a subset of rows from the public generator matrix with an error-free right hand side, and solving them. Let us define $\mathsf{t}$ as the ``number of errors (ones)'' in $\mathbf{e}$, and accordingly in $\mathbf{f}$  (see Fig.~\ref{FNError} showing that $\mathbf{e}$ and $\mathbf{f}$ include the same number of ones in their first $\mathsf{n}$ components). Inclusion  of $\mathsf{m}$ zeros at the end of the error vector $\mathbf{f}$ provides attacker with $\mathsf{m}$ error free equations in an information set decoding attack. This factor is accounted for in computations. Attacker first performs an exhaustive search over location of removed columns, and then in each case, conducts a second exhaustive search to find the set of equations to be used in information set decoding.  To conduct the second phase in exhaustive search, the attacker requires finding 
$\mathsf{d}-\mathsf{m}$ error free equations (or $\hat{\mathsf{d}}-\mathsf{m}$ considering the effect of removed columns) among 
$\mathsf{n}$ equations. This point is included in deriving the security levels throughout this work. 

Let us rely on notations prior to discarding of columns, e.g., $\mathsf{d}$ instead of $\hat{\mathsf{d}}$. To compute the probability of locating $\mathsf{d-m}$ error free equations among a total of $\mathsf{n}$, where $\mathsf{t}$ errors are included in the sequence of $\mathsf{n}$ bits, one can use the product of conditional probabilities, i.e.,
\begin{equation}
\left(\frac{\mathsf{n}-\mathsf{t}}{\mathsf{n}}\right)\!\!
\left(\frac{\mathsf{n}-\mathsf{t}-1}{\mathsf{n}-1}\right)\!..\!
\left(\frac{\mathsf{n}-\mathsf{t}-\mathsf{d}+\mathsf{m}+1}{\mathsf{n}-\mathsf{d}+\mathsf{m}+1}\right)\!.
\label{RR2Ne}
\end{equation}
We refer to $-\log_2$ of quantity in \ref{RR2Ne} as ``entropy of information set decoding''. Since Bob discards $\mathsf{p}$ randomly selected columns, in all our numerical computations, entropy of information set decoding, $\mathsf{d}$, is replaced by $\mathsf{d}-\mathsf{p}$, and 
$\mathsf{E}$ defined in~\ref{additive} is added to the result (also see expression \ref{AE4AK00abcde} and its relevant explanations). 
  
After~\cite{R0}, many research works on McEliece cryptosystem have focused on finding techniques that could speedup the information set decoding attack. This has resulted in a plethora of excellent research contributions, resulting in various levels of reduction in ``security level" vs.  ``attack complexity". Our comparisons are based on a select subset of such results published in~\cite{Main-ref}. Results corresponding to 
Figs.~\ref{Cross260},\ref{Cross311},\ref{Cross512},\ref{Cross1000} are shown in rows 1, 2, 3, 4 of the Table~\ref{TabN1AK}, respectively.

    \begin{figure}[h]
   \centering
\hspace*{-0.5cm}
   \includegraphics[width=0.485\textwidth]{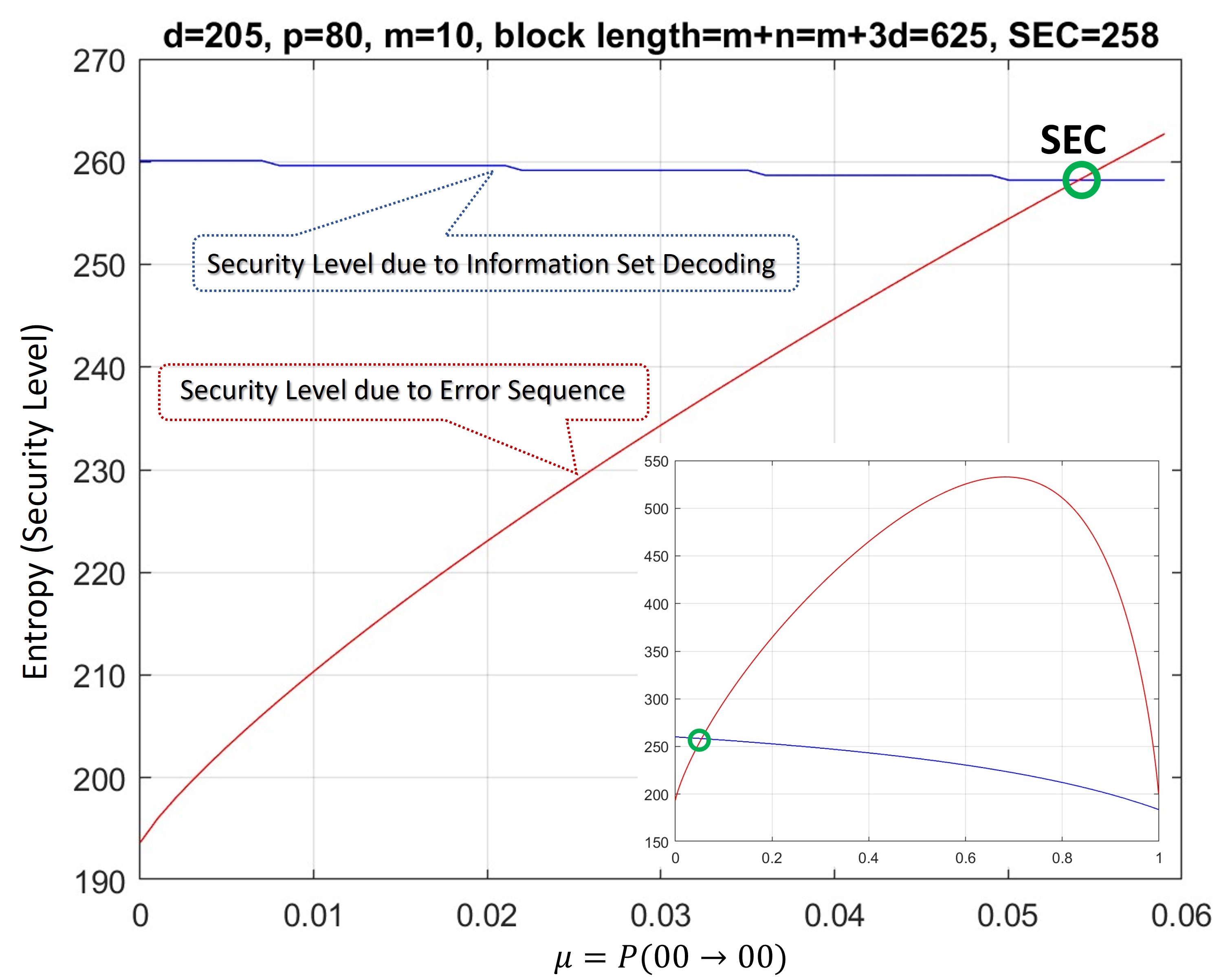}
   \caption{Example for security level due to the proposed method: 
$\mathsf{d}=205$, $\mathsf{p}=80$, $\mathsf{m}=10$, block length=625. Security level is 258 bits (circled point). Size of public key is $\mathsf{(n+m)d}=0.128$\,Mbits.}
\label{Cross260} 
 \end{figure}
    \begin{figure}[h]
   \centering
\hspace*{-0.45cm}
   \includegraphics[width=0.485\textwidth]{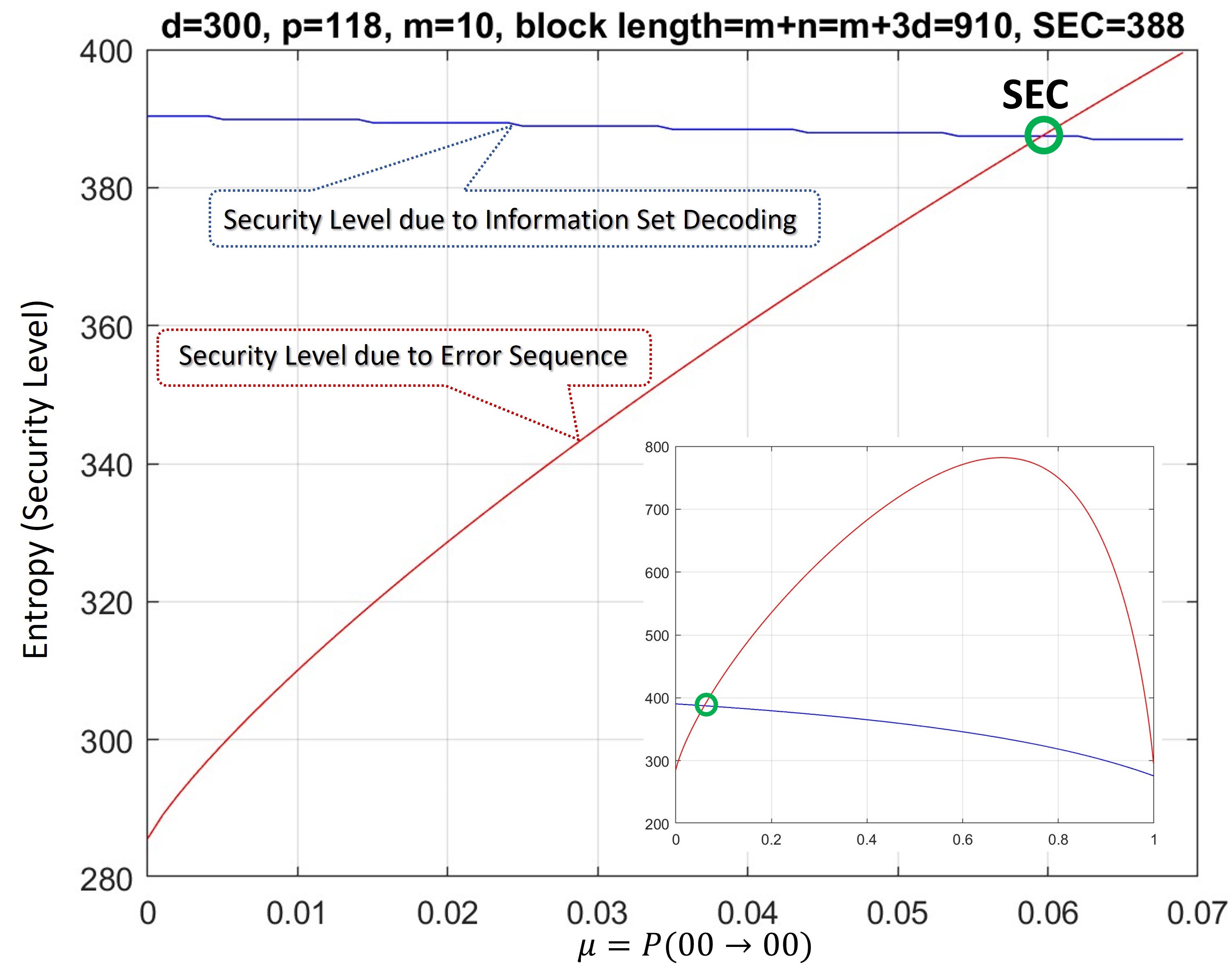}
   \caption{Example for security level due to the proposed method: 
$\mathsf{d}=300$, $\mathsf{p}=118$, $\mathsf{m}=10$, block length=910. Security level is 388 bits (circled point). Size of public key is $\mathsf{(n+m)d}=0.273$\,Mbits.}
\label{Cross311} 
 \end{figure}
    \begin{figure}[h]
   \centering
\hspace*{-0.5cm}
   \includegraphics[width=0.485\textwidth]{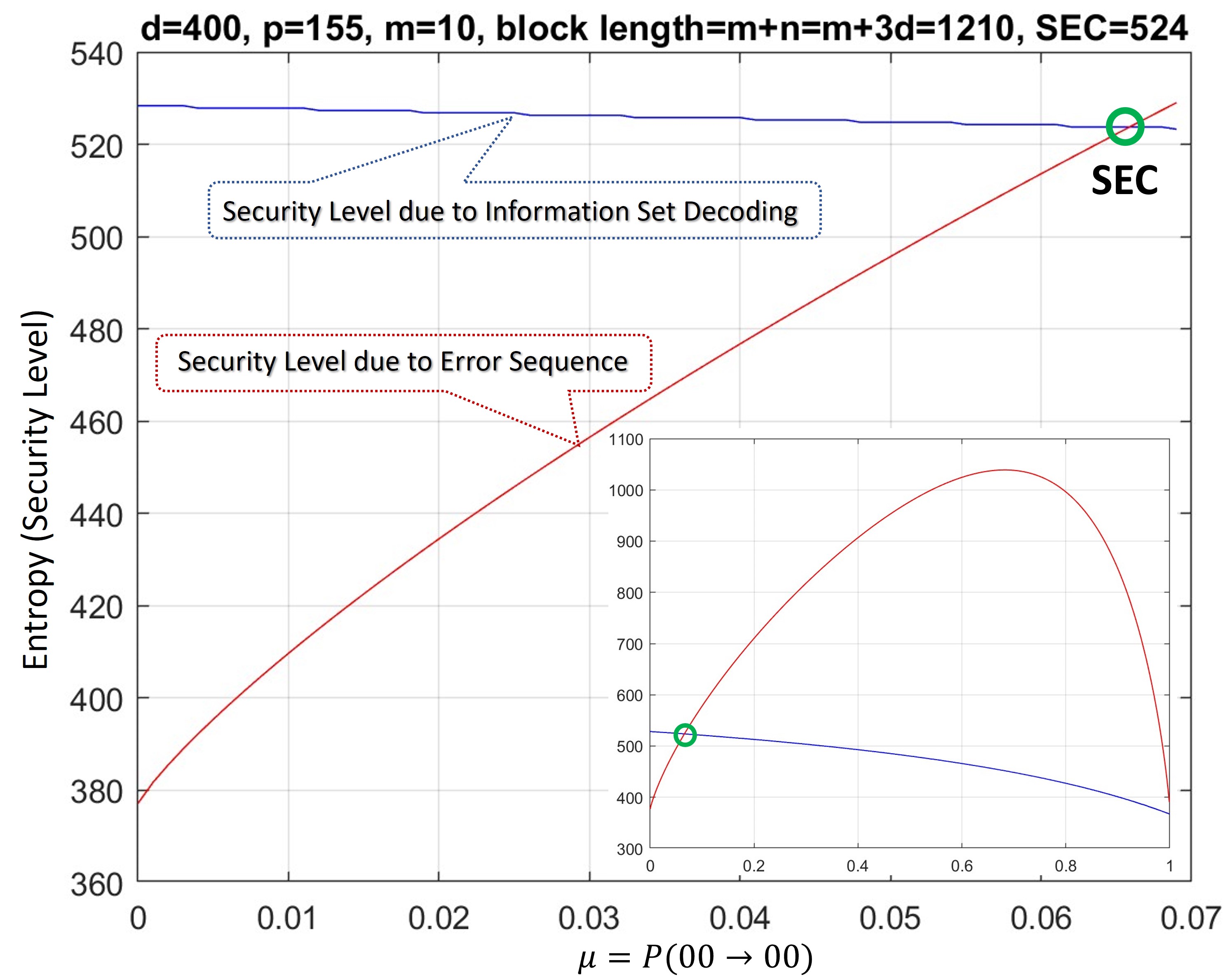}
   \caption{Example for security level due to the proposed method: 
$\mathsf{d}=400$, $\mathsf{p}=155$, $\mathsf{m}=10$, block length=1210.  Security level is 524 bits (circled point). Size of public key is $\mathsf{(n+m)d}=0.484$\,Mbits.}
\label{Cross512} 
 \end{figure}
    \begin{figure}[h]
   \centering
\hspace*{-0.45cm}
   \includegraphics[width=0.485\textwidth]{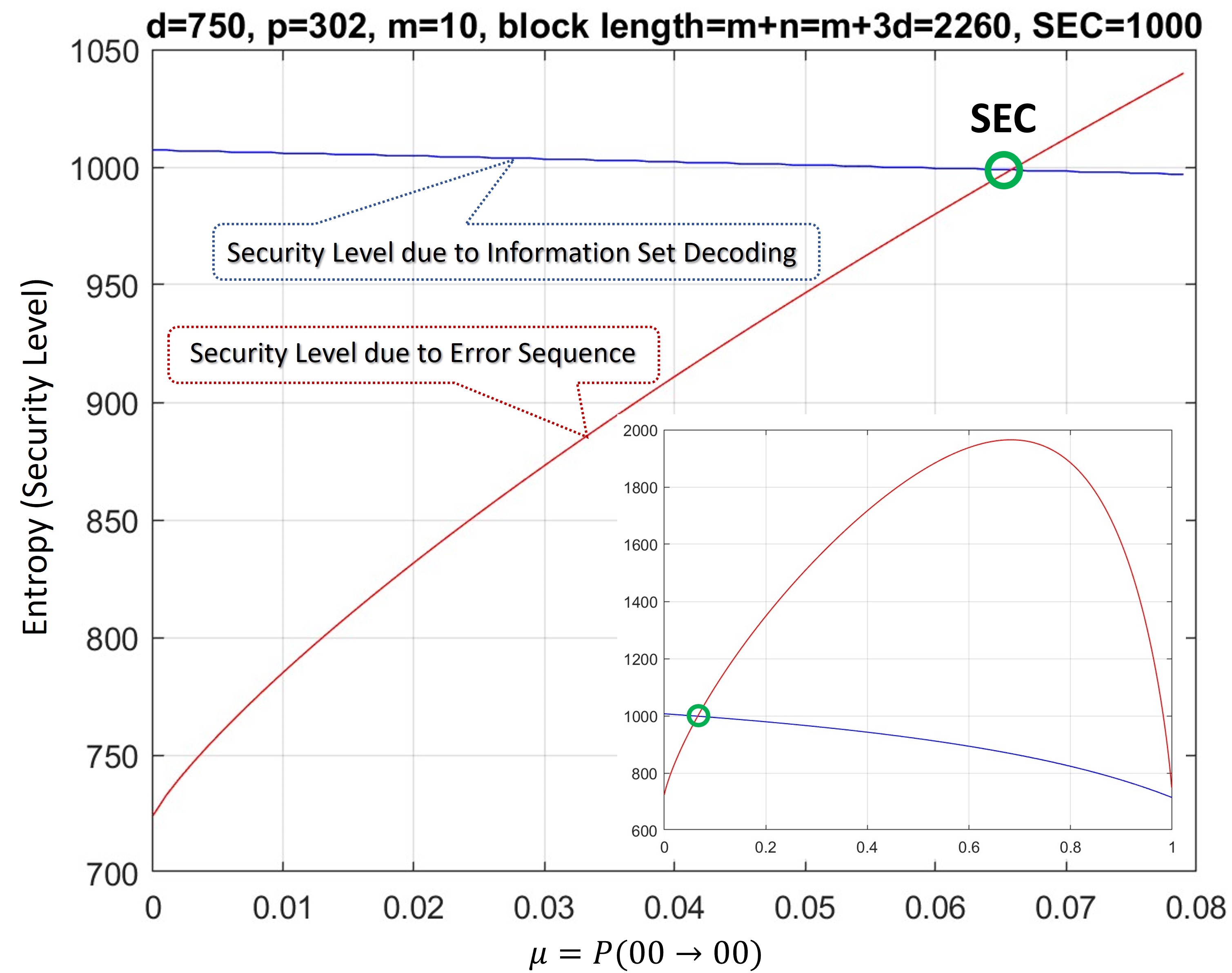}
   \caption{Example for security level due to the proposed method: 
$\mathsf{d}=750$, $\mathsf{p}=302$, $\mathsf{m}=10$, block length=2260.  Security level is 1000 bits (circled point). Size of public key is $\mathsf{(n+m)d}=1.695$\,Mbits.}
\label{Cross1000} 
 \end{figure}

\subsubsection{Verification Procedure in the Proposed Method}   
 
To verify if a  given iteration of information set decoding attack has resulted in the valid key, Eve should: (i) correctly guess which $\mathsf{p}$ columns are discarded,  (ii) find $\mathsf{d}-\mathsf{m}-\mathsf{p}$ equations (with an error-free right hand side) among  $\mathsf{n}$ equations, solve them, and then (iii) check the validity of the outcome by verifying that the corresponding syndrome is zero (verification).  Verification uses the parity check matrix obtained through steps (i), (ii); and then verifies if the candidate key results in a zero syndrome.  Making step 3 more complex is another contribution of this work, explained next. 

In the application of information set decoding attack to McEliece cryptosystem, since the number of errors  is known, one could limit the checking of the syndrome to those cases that the number of detected errors matches the known value. This two-steps verification would save the relatively computational intensive operation of checking the syndrome in each attempt. In most cases, by counting the number of detected errors, attacker can even avoid the need for checking of the syndrome.  In the proposed method, neither the number of inserted errors is known, nor one could compute the syndrome without knowing the positions of discarded columns.

{\bf Remark 1}: In this work, it is assumed that the validity of a potential key cannot be verified by decryption of a message. In other words, it is assumed the encrypted messages appear as a sequence of random bits. To hide any dependencies among bits, one could use a hash function to randomize all bits, append the structure of the hash function itself as a preamble in the resulting bit sequence, and then use the established key to encrypt the combination. The hash function can be changed for each encrypted file, and it can be simply a permutation of bits over sub-blocks of a relatively long length, where the seed used in generating the permutation is included in the permuted sequence of bits prior to encryption.  One can also rely on a nested encryption, where message is encrypted using a random key, then the key is included in the message preamble, and the outcome is encrypted once again using a second key which is established using techniques proposed here.  Randomizing message bits for each message  will also act as a mechanism to counter chosen cipher-text attack~\cite{RN80cip}. 
$\blacksquare$

Let us rely on using the notation $\hat{.}$ to specify discarded columns. 
In the proposed method, the number of added errors is not fixed. However, a valid error sequence would satisfy the constraints imposed by the state diagram in Fig.~\ref{FigS}. Then, verification could potentially start by examining if the error event is a valid sequence corresponding to a valid path in Fig.~\ref{FigS}. Let us consider an iteration of the information set decoding attack, and assume there is an error in solving the selected subset of equations. Let us assume there is at least a single bit in  the estimate of 
$\hat{\mathbf{d}}$ in error. Upon multiplication by 
$\hat{\mathbf{G}}$ in Fig.~\ref{FNSet2}, this results in selecting a column of 
$\hat{\mathbf{G}}_1$, and thereby a column of $\hat{\mathbf{G}}_2$, which in turn selects a linear combination of columns of matrix $\hat{\mathbf{D}}$. 
Note that $\hat{\mathbf{D}}$ and $\hat{\mathbf{G}}_2$ are composed of Independent and Identically Distributed (i.i.d.) bits with a probability 1/2 for zero and one, called a maximum entropy sequence hereafter. According to Theorem~\ref{Th2} in Appendix \ref{Theo1}, when such a vector is added to a vector due to other bits in $\hat{\mathbf{d}}$, it results in a maximum entropy vector. Attacker would examine the resulting vector to see if it could have been generated by the state diagram in Fig.~\ref{FigS}. In some cases, the result of this first step verification is an immediate rejection, since two ones follow each other indicating that the error vector is not generated by the state diagram. In cases that the error vector satisfies the state diagram, expression~\ref{RR3AK22p} shows that the probability of such a vector to be the actual error vector is negligible. This means relying on the result of the first step verification only marginally reduces the complexity in the overall verification. 

Consider the event $\mathcal{E}$ as the event that the detected error vector of length $\mathsf{n}$ satisfies the conditions imposed by the state diagram in Fig.~\ref{FigS}, and a second event that the detected vector is the actual error vector. If the entropy of the state diagram is equal to $\mathbb{H}_{s}(\mu)$, let us define $\mathcal{W}\in[0,1]$ as the indicator of the event that, conditioned on 
$\mathcal{E}$, the detected error sequence is the actual error. 
For large $\mathsf{n}$,  valid error vectors fall within a typical set composed of 
equiprobable sequences with a total probability approaching one \cite{Gallger}. 
The total number of sequences forming the typical set is equal to $2^{\mathsf{n}\mathbb{H}_{s}(\mu)}$ \cite{Gallger}. In other words, 
relying on the concept of typicality in Information Theory, for  a sufficiently large\footnote{The asymptotic result requires $\mathsf{n}\rightarrow\infty$, but values of $\mathsf{n}$ used here are large enough for the conclusions to be approximately valid.}  $\mathsf{n}$, any sequence generated by the state diagram, with a probability close to one, belongs to a typical set composed of $2^{\mathsf{n}\mathbb{H}_{s}(\mu)}$ equiprobable elements~\cite{Gallger}. As a result,   
\begin{equation}
P(\mathcal{W}=1|\mathcal{E}) \approx 2^{-\mathsf{n}\mathbb{H}_{s}(\mu)}.
\label{RR3AK22p}
\end{equation}
Note that in Figs.~\ref{Cross260},\ref{Cross311},\ref{Cross512},\ref{Cross1000} values of $\mathsf{n}$ are quite large,  and consequently, one can rely on expression \ref{RR3AK22p}. In Figs.~\ref{Cross260},\ref{Cross311},\ref{Cross512},\ref{Cross1000} the circled point corresponds to $\mu\approx 0.055$ and from Fig.~\ref{FigPEE} we have $\mathbb{H}_{s}(0.055)\approx 0.1$ bits. 
 Replacing in~\ref{RR3AK22p}, one can conclude that, if a sequence satisfies the restrictions imposed by the state diagram, it will be the actual error vector with a negligible probability. This negligible probability severely limits the attacker in simplifying the  information set decoding. As mentioned earlier, this is unlike McEliece cryptosystem in which the number of introduced errors (weight of error vector) is limited to the error correction capability of the code. This information can be exploited by an attacker to discard any outcome of the information set decoding which does not satisfy the known error weight, simplifying the attack.  

\section{Efficient Use of Random Matrices in Hiding Information} 

Referring to Appendix~\ref{mac}, in McEliece cryptosystem, the use of matrices $\mathbi{P}$ and $\mathbi{A}$ results in a generator matrix  
$\mathbi{PGA}$ which, unlike $\mathbi{G}$, includes randomness and thereby 
will be difficult to decode. However,  $\mathbi{PGA}$ generates the same set of code-words as $\mathbi{G}$, with some rearrangement of coordinates and reordering in the assignment of code-words to encrypted messages/keys. This feature causes two shortcomings: (i) Since the structure of the code, expect for some linear transformations, has remained the same, new decoding methods can be discovered/enhanced over time, which can decode the code $\mathbi{PGA}$ (see \cite{R4} to \cite{R11} as examples for decoding methods of a general linear code). Examples include methods based on iterative decoding, trellis representation with reduced complexity, and decomposition into cycle free structures, etc.  (ii) Information set decoding is simplified since attacker can reject candidates by simply counting the number of discovered errors. 

The method followed in this work enables hiding the generator matrix through addition of random masks. This means the modified (hidden) generator matrix does not generate the same set of code-words. This is explained next. 

\subsection{Masking through Multiplication of Matrices}
This work relies on appending random rows and random columns to matrices that are multiplied. This generates a binary matrix that will be (bit-wise) added to a matrix that should be kept secret. For example, referring to Fig.~\ref{FNSet2}, in the generation of public key, i.e.,  $\mathbf{BG}$, matrix $\mathbf{M=D}\mathbf{G}_2$ has masked the part that includes the original generator matrix, i.e.,   $\mathbf{C}\mathbf{G}_1$. 

\subsection{Maximizing the Entropy of the Product of Matrices} 

Using generic notations, Fig.~\ref{AFF1} shows the general form for a matrix multiplication, where  columns forming $\mathbf{Z}$ and  rows forming 
$\mathbf{U}$ are multiplied and the result, i.e., $\mathbf{Z}\mathbf{U}$, is added to $\mathbf{ST}$, thereby hiding its content. 
This section discusses how $\mathbf{Z}$ and $\mathbf{U}$ should be selected to maximize the entropy of $\mathbf{Z}\mathbf{U}$ (mask) and then the corresponding entropy is computed. 

\noindent
\underline{\bf Note}: In the following, in dealing with $\mathbf{Z}$, the term ``full rank'' means column-wise, i.e., columns of $\mathbf{Z}$ form a basis for $\mathbf{Z}$, and in terms of $\mathbf{U}$, it means row-wise, i.e., rows of $\mathbf{U}$ form a basis for $\mathbf{U}^t$, the transpose of $\mathbf{U}$. 
  \begin{figure}[h]
   \centering
\hspace*{-0.4cm}
   \includegraphics[width=0.5\textwidth]{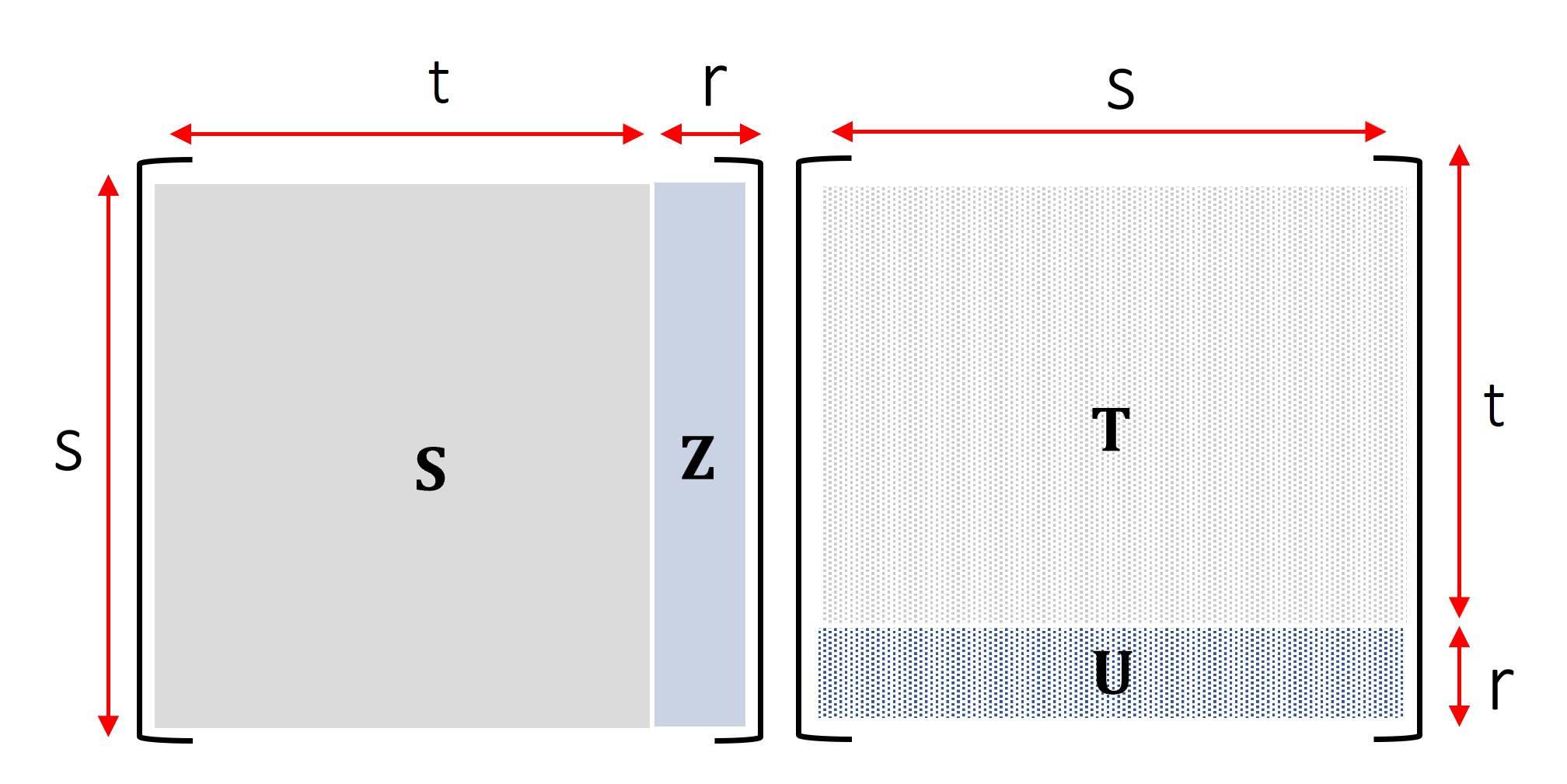}
   \caption{Formation of masking towards deriving the conditions for achieving  maximum entropy. Note that here, the result of the multiplication is 
$\mathbf{ST}+\mathbf{ZU}$, where $\mathbf{ST}$ is a matrix that should be kept secret, e.g., the generator matrix used in key encapsulation. This means $\mathbf{ZU}$ has masked the product $\mathbf{ST}$.}
   \label{AFF1}
 \end{figure}
Let us represent rows of $\mathbf{Z}$ by $\mathbf{r}_i, i=1,\ldots,
=\mathsf{s=r+t}$ and columns of $\mathbf{U}$ by $\mathsf{c}_i, i=1,\ldots,\mathsf{s=r+t}$. We have,
\begin{equation}
\mathbf{Z}\mathbf{U}=\sum_{i=1}^{\mathsf{s=r+t}}\mathbf{c}_i\circledcirc \mathbf{r}_i
\label{Eq2}
\end{equation}
where $\circledcirc$ specifies the multiplication of a column vector by a row vector resulting in an $\mathsf{s}\times \mathsf{s}$ matrix which takes different random realizations. 

Hereafter, a matrix of the form $\mathbf{Z}\mathbf{U}$ is referred to as a masking matrix, and its components of the form $\mathsf{c}_i\circledcirc \mathsf{r}_i$ in \ref{Eq2} are referred to as atomic masks.  A realization for each atomic mask are added to form a realization of $\mathbf{Z}\mathbf{U}$. 
Atomic mask $\mathbf{c}_{\kappa}\circledcirc \mathbf{r}_{\kappa}$
 is obtained by repeating column of $\mathbf{Z}$ indexed by ${\kappa}$ at positions that contain a one in the row of $\mathbf{U}$ indexed by 
${\kappa}$. Obviously, any change in $\mathbf{c}_{\kappa}$ and/or in 
$\mathbf{r}_{\kappa}$ results in a new realization for $\mathbf{c}_{\kappa}\circledcirc \mathbf{r}_{\kappa}$. In the formation of the atomic mask $\mathbf{c}_{\kappa}\circledcirc \mathbf{r}_{\kappa}$, there are $2^\mathsf{s}$ possibilities for $\mathbf{c}_{\kappa}$ and $2^\mathsf{s}$ possibilities for $\mathbf{r}_{\kappa}$, resulting in $2^{2\mathsf{s}}$ single elements in $\mathbf{c}_{\kappa}\circledcirc \mathbf{r}_{\kappa}$, while some of these single elements are equal to each other. For example, if $\mathbf{c}_{\kappa}$ (or  
$\mathbf{r}_{\kappa}$) is zero, it will result in an all-zero atomic mask regardless of the value taken by the other component, i.e., $\mathbf{r}_{\kappa}$  (or $\mathbf{c}_{\kappa}$), respectively. Such repetitions are avoided if the all zero vector is not included in realizations used for generating atomic masks. A second possibility for repetition occurs if, in some realizations of atomic masks, sum of multiple atomic masks add up to zero. 
Theorem~\ref{fullrank} show (with exceptions explained in Remark 4) both these requirements are satisfied iff realizations of  $\mathbf{Z}$ and realizations of $\mathbf{U}$ used in generating a mask are full-rank.

\begin{theorem} \label{fullrank}
 Condition of maximum entropy is satisfied  iff $\mathbf{Z}$ and $\mathbf{U}$ are selected with uniform probabilities over all elements in their respective sets of full-rank realizations (with exceptions explained in Remark 4).    
\end{theorem}

\begin{proof}
See Appendix~\ref{ProFR}.
\end{proof}

\begin{theorem} \label{fullrank2}
Subject to conditions of Theorem \ref{fullrank}, let us consider a set of columns and a set of rows in the generated mask such that their intersections form a column-wise as well as a row-wise basis, i.e., span an $\mathsf{r}\times \mathsf{r}$ binary space. Then, the rectangular  matrix formed at the intersections  is equally likely to be any invertible binary matrix of size $\mathsf{r}\times \mathsf{r}$. 
\end{theorem}

\begin{proof}
See Appendix~\ref{ProFR2}.
\end{proof}

{\bf Remark 2:} 
Marginal distribution of each column of $\mathbf{Z}$ and each row of $\mathbf{U}$  is composed of $2^{\mathsf{s}}-1$ vectors of equal probabilities. However, due to the constraint on matrix rank, the vectors cannot be selected independently. 
$\blacksquare$

{\bf Remark 3:} If the full rank condition explained in Theorem~\ref{fullrank} is not satisfied, we obtain repeated terms resulting in non-equal probabilities and the overall entropy reduces (see Theorem~2.3.1 of \cite{Gallger}).   All-zero masks are avoided (by avoiding all-zero columns in $\mathbf{Z}$ and/or all-zero rows in $\mathbf{U}$), since, if permitted, it would result in all-zero atomic mask formed from each such all-zero vector. This would result in a high probability for the all-zero mask which lacks information hiding property.  $\blacksquare$

\subsection{Number of Distinct Masks}

As indicated in \ref{Eq2} and Fig.~\ref{AFF1}, the products
$\mathbf{Z}\mathbf{U}$ and $\mathbf{S}\mathbf{T}$ are of size 
$\mathsf{s}\times \mathsf{s}$.  
Matrices $\mathbf{Z}$ and 
$\mathbf{U}$ are  composed of $\mathsf{r}$ columns of size $\mathsf{s}$ and $\mathsf{r}$ rows of size $\mathsf{s}$, respectively. Let us focus on matrix $\mathbf{Z}$. 

Consider a full rank, random binary matrix $\mathbf{Z}$ of size $\mathsf{s}\times \mathsf{r}$ with $\mathsf{s}>\mathsf{r}$. For columns of $\mathbf{Z}$ to span $\mathsf{r}$ dimensions, the first column, $\jmath=0$, can take all $2^{\mathsf{s}}$ combinations except the all zero vector, and columns $\imath=1,...,\mathsf{r}-1$ cannot be a linear combination of columns   
$\jmath=0,...,\imath-2$. Putting these arguments together, it is concluded that for a given $\imath=0,...,\mathsf{r}-1$, the number of such linear combinations to be excluded is equal to $2^{\imath}$. As a result, the total number of ways for selecting matrix $\mathbf{Z}$ is equal to:
\begin{equation}
\hbar=\prod_{i=0}^{\mathsf{r}-1}(2^{\mathsf{s}}-2^{i})
=\prod_{i=0}^{\mathsf{r}-1}2^{\mathsf{s}}(1-2^{i-{\mathsf{s}}}).
\label{Eq14AK}
\end{equation}
\begin{equation}
\log_2(\hbar)=\mathsf{r}\mathsf{s}+\sum_{i=0}^{\mathsf{r}-1}
\log_2(1-2^{i-\mathsf{s}})
\label{LogAK1p}
\end{equation}
since $\mathsf{z}\equiv (1-2^{i-\mathsf{s}})>0$, we can use the well known inequality~\cite{Gallger} (expression 2.3.2 in Theorem 2.3.1), 
\begin{equation}
\ln \mathsf{z} \leq (\mathsf{z}-1),~\forall \mathsf{z}>0.
\label{LogAK13}
\end{equation}
Substituting \ref{LogAK1p} in \ref{LogAK13} results in the following upper bound
\begin{eqnarray} \nonumber 
\log_2(\hbar)& \leq & \mathsf{r}\mathsf{s}-\log_2(e)\sum_{i=0}^{\mathsf{r}-1} 2^{i-\mathsf{s}}\\   \label{LogAK3}
& \leq & \mathsf{r}\mathsf{s}+\log_2(e)2^{-\mathsf{s}}\left(1-2^{\mathsf{r}}\right) 
\end{eqnarray}
where we have used: $\sum_{i=0}^{\mathsf{r}-1} 2^{i}=(2^{\mathsf{r}}-1)$.
Defining, $\mathsf{u}=1/\mathsf{z}>0$ results in, 
\begin{equation}
\ln \mathsf{u} \geq \left(1-\frac{1}{\mathsf{u}}\right),
\end{equation}
or equivalently, 
\begin{equation}
\log_2(\mathsf{u}) \geq \log_2(e)\left(1-\frac{1}{\mathsf{u}}\right).
\label{LogAK2pp}
\end{equation} 
Considering $\mathsf{u}=1-2^{i-\mathsf{s}}$, expressions \ref{LogAK1p} and \ref{LogAK2pp} can be combined, resulting in the following lower bound,
\begin{eqnarray} \nonumber 
\log_2(\hbar)& \geq & \mathsf{r}\mathsf{s}+\log_2(e)\sum_{i=0}^{\mathsf{r}-1}
\left(1-\frac{1}{1-2^{i-\mathsf{s}}}\right)\\  \label{LogAK3pppp}
\log_2(\hbar)& \geq & \mathsf{r}\mathsf{s}-\log_2(e)
\sum_{i=0}^{\mathsf{r}-1}
\frac{2^{i-\mathsf{s}}}{1-2^{i-\mathsf{s}}} \\   \label{LogAK3ppp}
\log_2(\hbar)& \geq & \mathsf{r}\mathsf{s}-
\log_2(e)
\left(\frac{\mathsf{r}2^{\mathsf{r}-1-\mathsf{s}}}{1-2^{\mathsf{r}-1-\mathsf{s}}}\right),
\end{eqnarray}
where \ref{LogAK3ppp} is concluded from \ref{LogAK3pppp} noting that 
\begin{equation}
\frac{2^{i-\mathsf{s}}}{1-2^{i-\mathsf{s}}}
\end{equation} 
is an increasing function of $i$. 
The gap between lower and upper bounds is computed by subtracting the right hand side of  \ref{LogAK3} from the right hand side of 
\ref{LogAK3ppp}, resulting in:
\begin{eqnarray}
\log_2(e)\left(2^{-\mathsf{s}}-2^{\mathsf{r}-\mathsf{s}}+
\frac{\mathsf{r}2^{\mathsf{r}-1-\mathsf{s}}}{1-2^{\mathsf{r}-1-\mathsf{s}}}\right). \label{LogAK3z}
\end{eqnarray}
Figure~\ref{FNAKK} shows examples of the values of the gaps, where 
$\mathsf{s}$ and $\mathsf{r}$ are selected to provide estimates of the gap. It is observed that the gap values are negligible. In practice, values of $\mathsf{s}$ (limiting the gap) are significantly larger. 

    \begin{figure}[h]
   \centering
\hspace*{-0cm}
   \includegraphics[width=0.5\textwidth]{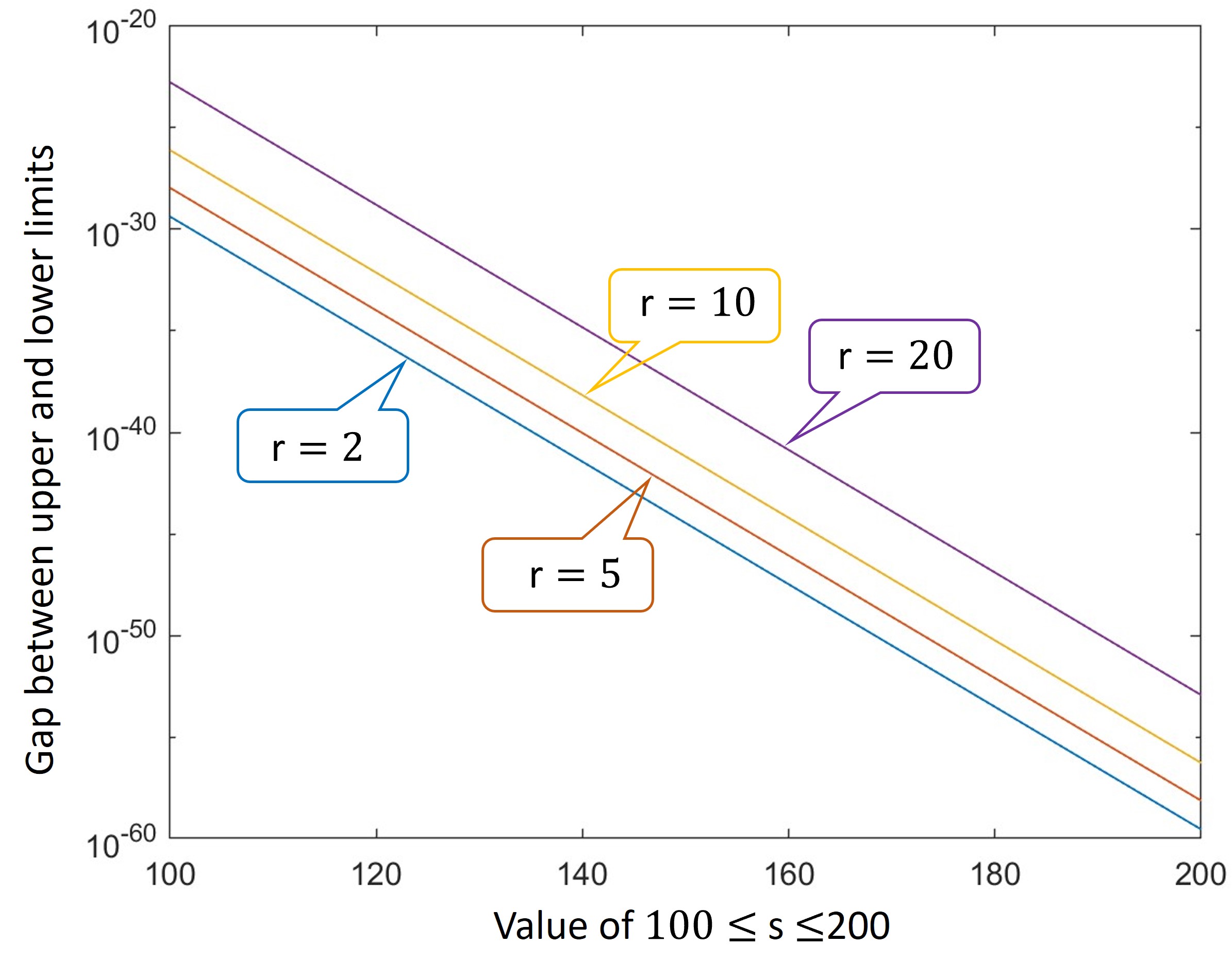}
   \caption{Values of the gap in~\ref{LogAK3z} as a function of $\mathsf{s}$ for different values of $\mathsf{r}$.} 
\label{FNAKK} 
 \end{figure} 

 Since the gap between bounds is negligible (see Fig.~\ref{FNAKK}), in what follows, we simply rely on \ref{LogAK3ppp} to represent an estimate of $\log_2(\hbar)$. We have
\begin{equation} \label{AKQWER}
\mathsf{r}\mathsf{s}-
\log_2(e)
\left(\frac{\mathsf{r}2^{\mathsf{r}-1-\mathsf{s}}}{1-2^{\mathsf{r}-1-\mathsf{s}}}\right)\approx \mathsf{r}\mathsf{s},
\end{equation}
for values of $\mathsf{r}$ and $\mathsf{s}$ used in practice. In numerical results provided in Table~\ref{TabN1AK}, the generic notation $\mathsf{s}$ and $\mathsf{r}$ correspond to $\mathsf{n+m}$ and $\mathsf{m}$, respectively. In practice, $\mathsf{n}$ is selected to be quite larger than  $\mathsf{m}$. This means, in \ref{AKQWER}, the quantity 
$2^{\mathsf{r}-1-\mathsf{s}}\approx 0$, concluding the final approximation, i.e., $\approx \mathsf{r}\mathsf{s}$ in  \ref{AKQWER}.  
Similar arguments can be applied to compute  the number of matrices 
$\mathbf{U}$. 

{\bf Remark 4:} Even if the condition of relying on full rank realizations is followed, there will be repeated masking matrices if  realizations of $\mathbf{c}_i\circledcirc \mathbf{r}_i$ in \ref{Eq2} occur in a different order. This means, if 
columns in a realization of $\mathbf{Z}$ and rows in a realization of
$\mathbf{U}$ are permuted (using the same permutation for both), the expression in \ref{Eq2} results in the same final mask. In this case, for each masking matrix there will be multiple identical copies that simplify the task of exhaustive search by limiting the verification to only one of the identical copies. 
This factor is accounted for by using ${\mathsf{r}}!$ in expression \ref{Eq15}.  
We refer to the subset  of masking matrices composed of a single element as a {\em Ground Set}. There will be ${\mathsf{r}}!$ identical repetitions of the Ground Set. This means there is a partition of the larger set into ${\mathsf{r}}!$ equivalent subsets (any of the repeated subsets can be considered as the  Ground Set). Moving forward, the article  focuses on using a single Ground Set. The conditions for maximizing entropy will not be violated, since all elements in the collection of any given equivalent subsets occur with the same probability. 
$\blacksquare$

Accounting for the repeated elements due to permutation, the total  number of distinct masking matrices in a Ground Set due to $\mathbf{ZU}$ will be
\begin{equation}
\frac{\hbar^2}{{\mathsf{r}!}}\approx\frac{2^{2\mathsf{r}\mathsf{s}}}{\mathsf{r}!}.
\label{Eq15}
\end{equation}
For typical values of $\mathsf{r}$ and $\mathsf{s}$ (selected to hide the key in an efficient manner as summarized in Table~\ref{TabN1AK}), the entropy of the term in \ref{Eq15}, i.e., 
\begin{equation}
2\log_2(\hbar)-\log_2(\mathsf{r}!)\approx 2\mathsf{r}\mathsf{s}-\log_2(\mathsf{r}!)
\label{Eq15AK}
\end{equation}
is much larger than the key entropy and relevant $\mathsf{SEC}$ value 
($\mathsf{SEC}$ is defined in expression \ref{AE4AK00}). This entails an attacker relying on exhaustive search can handle the search over these factors easier than trying to first exhaustively find the masking matrix.

{\bf Remark 5:} Note that each column in $\mathbf{Z}$, and similarly each row in $\mathbf{U}$, can take all possible (non-zero) binary combinations, i.e.,  $2^{\mathsf{s}}-1$ values. This point will be later used in some of  the proofs.  
$\blacksquare$

{\bf Remark 6:} In case that matrices $\mathbf{Z}$ and $\mathbf{U}$ are not rectangular, e.g.,  $\mathbf{Z}$ and $\mathbf{U}$ are of sizes $\mathsf{s}_1\times \mathsf{r}$ and $\mathsf{r}\times \mathsf{s}_2$, respectively, the right hand side in expression  \ref{Eq15AK} changes to:
\begin{equation}
\text{Mask Entropy}\approx \mathsf{r}(\mathsf{s}_1+\mathsf{s}_2)-\log_2(\mathsf{r}!).
\label{Eq15AKKH}
\end{equation}
For typical values of parameters, the quantity in \ref{Eq15AKKH} is larger than key entropy and relevant $\mathsf{SEC}$ ($\mathsf{SEC}$ is defined in expression \ref{AE4AK00}). $\blacksquare$

\subsection{Information Hiding by Incorporating Restricted Randomness}\label{RES1}

First, two simple definitions: We refer to a vector/matrix composed of i.i.d. binary values of probabilities 1/2 for \{0,1\} as ``truly random, or unrestricted random", to a vector/matrix composed of binary values with statistical dependencies among them as ``partially random, or partially restricted" and to binary values completely determined by the rest of the bits as ``fully restricted, or deterministic".

To enable key recovery, Alice should follow a particular structure for incorporating randomness shown (see Fig.~\ref{FNSet1}) while satisfying restrictions shown in Figs.~\ref{FigA},\ref{FNB1},\ref{FNAB},\ref{FNSet2}.   These restrictions reduce the entropy of matrix ${\mathbf B}$, which plays the main role in hiding the secret data embedded in the public key,  as compared to its size in bits, i.e.,  ${\mathbf B}$ is only partially random. Bob introduces randomness by discarding $\mathsf{p}$  randomly selected columns from ${\mathbf P}$ to generate $\hat{\mathbf P}_1$, which is then multiplied by a randomly selected data vector $\hat{\mathbf d}$ of length $\mathsf{d-p}$.  This means randomness introduced by Bob is unrestricted. 

\subsection{Introducing Randomness Supporting Key Recovery}\label{tofollow}

 In original McEliece  cryptosystem, randomness is introduced  by changing the generator matrix to an equivalent code which is difficult to decode. The aim is that  exhaustive search is the only option for an attacker to figure out randomness introduced into the key structure. {\em Information set decoding} is one such exhaustive search, followed by solving a set of linear equations to find the hidden key. In any exhaustive search, upon each new attempt, attacker needs to verify if the attack has been successful. This step is called "verification", hereafter. One should introduce the randomness in a manner that the verification is as difficult as possible.

To keep expressions simple and general, the effect of  discarding columns of the generator matrix (by Bob) is not implicitly included in this section. This means notation $\mathbf{P}$ in some cases represents reduced public generator, i.e.,  $\hat{\mathbf{P}}$ and $\mathbf{d}$ represents shortened data vector $\hat{\mathbf{d}}$. This will be clear from the context. 
 
Let us consider matrices $\mathbf{A}$ and $\mathbf{B}$ defined in
 Figs.~\ref{FigA},\ref{FNSet2},\ref{FNrecover},\ref{FNSet1}. Product $\mathbf{A}$ and $\mathbf{B}$ is of the form shown in Fig.~\ref{FNAB} (also see \ref{FNSet2}). Figure \ref{FNB2} shows a different decomposition of matrix  
$\mathbf{B}$, with sub-matrices  $\mathbf{C}_1$, $\mathbf{C}_2$, $\mathbf{D}_1$ and $\mathbf{D}_2$, that will be used in some of the derivations.  Sub-matrices  
$\mathbf{R}$ and $\mathbf{Q}$ in 
$\mathbf{A}$, and $\mathbf{D}$ (decomposed into $\mathbf{D}_1$ and $\mathbf{D}_2$ in Fig.~\ref{FNSet1}) in $\mathbf{B}$ can be selected unrestricted. The reason is that, for any realization of these matrices, the remaining sub-matrices can  be selected to allow key encapsulation/recovery, i.e., satisfy the condition on the product of    $\mathbf{AB}$ in Figs.~\ref{FNAB},\ref{FNSet2}. 
 
The first step in recovering the key is to multiply the received vector (including the added error vector $\mathbf{e}$) by the matrix $\mathbf{A}$. 
Added error vector $\mathbf{e}$ is of the form shown in Fig.~\ref{FNError}.  
Figure~\ref{FNrecover} shows the result of multiplication of the error vector $\mathbf{e}$ by $\mathbf{A}$. Since matrix $\mathbf{A}$ includes an identity matrix in its upper left sub-matrix, vector  $\mathbf{f}$ remains unchanged, while vector $\mathbf{Rf}$ (see Fig.~\ref{FigA} for definition of $\mathbf{R}$) will be discarded by Alice.

Referring to Fig.~\ref{FNSet1}, expressions (a) to (g) govern the outcome of the  matrix multiplication $\mathbf{AB}$, an operation which is implicitly performed at Alice's side by multiplying the received vector by $\mathbf{A}$. All expressions are straightforward, expect for (f). Using (a) in Fig.~\ref{FNSet1}, we have
\begin{equation}
\textrm{(a)}\rightarrow 
\mathbf{C}_2=\mathbf{Q}^{-1}\mathbf{R}\mathbf{C}_1.
\label{EXXX1}
\end{equation}
Replacing in (b), we obtain
\begin{equation}
\textrm{(a),(b)}\rightarrow \mathbf{C}_1=\mathbf{I}+\mathbf{S}\mathbf{Q}^{-1}\mathbf{R}\mathbf{C}_1
\label{EXXX2}
\end{equation}
or
\begin{equation}
\textrm{(a),(b)}\rightarrow \mathbf{C}_1(\mathbf{I}+\mathbf{S}\mathbf{Q}^{-1}\mathbf{R})=\mathbf{I}.
\label{EXXX3}
\end{equation}
Our goal is to determine which of the matrices involved can be selected  unrestricted. In establishing this property, we can select any subset of matrices unrestricted if, through restricting the remaining matrices, we are able to satisfy the expressions in (a) to (g) in Fig.~\ref{FNSet1}, thereby guaranteeing key recovery.   

Matrices $\mathbf{C}_1$  and $\mathbf{C}_2$ are partially restricted due to their  relationship captured in (a). Product of 
$\mathbf{S}$ and $\mathbf{C}_2$ is an $\mathsf{n}\times \mathsf{n}$ 
matrix, same size as  $\mathbf{C}_1$. However, the information content of $\mathbf{S}$, and likewise that of $\mathbf{C}_2$, is at most $\mathsf{mn}$ bits, adding to $\mathsf{2mn}$ bits in total. On the other hand, size of $\mathbf{C}_1$ is $\mathsf{n}^2$. For parameters used in practice of the proposed techniques, we have $\mathsf{n}^2\gg 2\mathsf{mn}$, since 
$\mathsf{n}\gg 2\mathsf{m}$. This means (a) in Fig.~\ref{FNSet1} includes partially restricted components. On the other hand, referring to (d) to (g), one can select 
$\mathbf{D}_1$, $\mathbf{D}_2$, $\mathbf{Q}$ and $\mathbf{R}$ unrestricted, while, noting (e) in Fig.~\ref{FNSet1}, matrix $\mathbf{S}$ will be deterministic for any realizations of $\mathbf{D}_1$ and  $\mathbf{D}_2$. Same conclusion can be reached by noting (a) where, given $\mathbf{C}_1$ and $\mathbf{C}_2$, matrix $\mathbf{S}$ will be determined. In  (b), one can  select $\mathbf{R}$ and $\mathbf{Q}$ unrestricted, and restrict $\mathbf{C}_1$, $\mathbf{C}_2$ such that (b) is satisfied. Likewise, in (c), one can freely select $\mathbf{D}_1$, $\mathbf{D}_2$ and restrict  $\mathbf{S}$ such that (c) is satisfied. For (d), all components on the right hand side are unrestricted, and as a result, $\mathbf{E}$ will be unrestricted as well. Finally, (f) and (g) capture the fact that $\mathbf{S}$ is fully restricted, while $\mathbf{C}_1$ and $\mathbf{C}_2$ are partially restricted. These conditions are consistent with what appears on the two sides of (f) and (g).

Another objective is to hide the generator matrix, in the sense that the information an eavesdropper requires in order to extract the generator matrix by observing public key is larger than the entropy of the key itself, as well as that of $\mathsf{SEC}$ (given in \ref{AE4AK00}). 
Noting Fig.~\ref{FNSet2},  one can conclude that the public key is equal to:
\begin{equation} \label{AD0}
\mathbf{B}\mathbf{G}=\mathbf{C}\mathbf{G}_1+\mathbf{D}\mathbf{G}_2.
\end{equation}
Since $\mathbf{D}$ and $\mathbf{G}_2$ are unrestricted,  according to Theorem~\ref{Th2} in Appendix \ref{Theo1}, the entropy of the masking matrix, i.e., 
\begin{equation} \label{AD1}
\mathbf{M}=\mathbf{D}\mathbf{G}_2
\end{equation}
will be at least
\begin{equation} \label{AD2}
\mathsf{m(m+n)+md}-\log_2(\mathsf{m}!)~\mbox{bits}.
\end{equation}
From~\cite{Stirling}, we have 
\begin{equation} \label{AD3}
\mathsf{m}!<\sqrt{2\pi \mathsf{m}}\left(\frac{\mathsf{m}}{e}\right)^\mathsf{m}\exp\left(\frac{1}{12\mathsf{m}+1}\right). 
\end{equation}
Replacing in \ref{AD3}, we conclude:
\begin{equation} \label{AD2}
\mathsf{m(m+n+d)}-\log_2(\mathsf{m}!)\gg \mathsf{SEC}
\end{equation}
where $\mathsf{SEC}$ is given in \ref{AE4AK00}. 


\subsection{Some Attack Strategies}

Next, let us consider  a few attack strategies:

\noindent {\bf 1:} Attacker targets the decomposition of $\mathbf{B}=\mathbf{C}\mathbf{G}_1+\mathbf{D}\mathbf{G}_2$ towards extracting information relevant to $\mathbf{C}\mathbf{G}_1$ and thereby relevant to generator $\mathbf{G}_1$.

\noindent {\bf 2:} Attacker targets the vector $\mathbf{B}\mathbf{G}\mathbf{d}+\mathbf{e}$ to extract information about $\mathbf{d}$. 

\noindent {\bf 3:}  Attacker targets matrix $\mathbf{A}$, which can then be used to extract the key. 

 \noindent Referring to Fig.~\ref{FNSet1}, and noting the formation of terms involved in above three cases, the entropy required to successfully conduct the above attacks is, for typical values of $\mathsf{n}$ and $\mathsf{m}$, substantially higher than the $\mathsf{SEC}$ value in \ref{AE4AK00}. 

\noindent{\bf 4:}  Referring to Fig.~\ref{FNSet1}, attacker may rely on public key matrix, i.e., matrix  $\mathbf{P=BG}$, expressed as
\begin{equation}
\mathbf{P}=\mathbf{BG}=\mathbf{C}\mathbf{G}_1+\mathbf{D}\mathbf{G}_2=\mathbf{C}\mathbf{G}_1+\mathbf{M},
\end{equation}
to extract information about generator $\mathbf{G}_1$. However, since $\mathbf{C}\mathbf{G}_1$, which contains all the information about $\mathbf{G}_1$, is masked by matrix $\mathbf{M}=\mathbf{D}\mathbf{G}_2$, which is an unrestricted matrix, the attacker's only strategy can be that of finding 
$\mathbf{M}=\mathbf{D}\mathbf{G}_2$ through an exhaustive search. Noting expression~\ref{Eq15AK}, the entropy of $\mathbf{M}=\mathbf{D}\mathbf{G}_2$ is substantially higher than the length of the key, as well as the $\mathsf{SEC}$ value (given in \ref{AE4AK00}).  Consequently, $\mathbf{C}\mathbf{G}_1$ is protected. 

\subsection{Computation of Security Level}
Note that verification process requires finding the following two items. 

\noindent 
{\bf Item 1} concerns  the location of discarded columns in $\mathbf{P}$ to obtain $\hat{\mathbf{P}}$. For item 1, the $\log_2$ of the number of possibilities is obtained using \ref{additive}. 

\noindent 
{\bf Item 2} concerns values of bits forming the shortened key of length $\mathsf{d-p}$. These are needed to compute the syndrome corresponding to the shortened key using parity generator matrix corresponding to  $\hat{\mathbf{P}}$, obtained in item 1 (also see Remark 1).  There are two options to cover Item 2: 

\noindent 
{\bf Option 1} is based on information set decoding. The $\log_2$ of the number of possibilities for information set decoding is computed using $\log_2$ of the quantity in 
\ref{RR2Ne}, i.e.,  ``entropy of information set decoding''. 
In computing \ref{RR2Ne}, $\mathsf{d}$ is replaced by $\hat{\mathsf{d}}=\mathsf{d-p}$ (number of required equations to find $\mathsf{d}$ for each possible case of selecting the discarded columns),  $\mathsf{m}$ is subtracted to account for $\mathsf{m}$ error-free equations publicly known, and $\mathsf{t}$ is computed as the average number of ones produced by the state diagram in Fig.~\ref{FigS} over a block of length $\mathsf{n}$, namely 
\begin{equation} \label{AE4AK00qwe}
\mathsf{t}=\mathsf{n}\pi_1(1-\mu)  =  \frac{\mathsf{3d(1-\mu)}}{3-2\mu}
\end{equation}
where $\mathsf{n}=3\mathsf{d}$ and $\pi_1$ given in \ref{Ent1AK}. Using average number of erroneous bits is not entirely accurate, however, noting law of large numbers, for values of $\mathsf{n}$ used in practice, it provides a fairly accurate approximation.  

\noindent 
{\bf Option 2} is based on finding the error vector $\mathbf{e}$, 
with an entropy of
\begin{equation} \label{AE4AK00qwert}
\text{Entropy of error vector}=\mathsf{n}\mathbb{H}_{s}(\mu)
\end{equation}
with $\mathbb{H}_{s}(\mu)$ given in \ref{Ent3}. 

Then, $\mathsf{SEC}$ is defined as
\begin{equation} \label{AE4AK00}
\mathsf{SEC}=\min(\mathcal{A},\mathcal{B}).
\end{equation}
Terms $\mathcal{A}$, $\mathcal{B}$ in \ref{AE4AK00} are each the summation of  two terms. 
For $\mathcal{A}$, the two terms are:
Item 1 plus Option 1 for computing Item 2.  
For $\mathcal{B}$, the two terms are:
Item 1 plus Option 2 for computing Item 2.  
The point corresponding to the intersection of the two curves, representing $\log_2(\mathcal{A})$ and 
$\log_2(\mathcal{B})$ (see Figs. \ref{Cross260} to \ref{Cross1000}, and  Table~\ref{TabN1AK}) is selected to maximize the
$\mathsf{SEC}=\min(\mathcal{A},\mathcal{B})$. 

\subsection{Direct Computation of the Key Entropy}

The final key is composed of 
$\hat{\mathsf{d}}$ random bits with $\mathsf{p}$ zeros inserted in the location of discarded bits (positions are random and known neither to Alice nor to the public). Let us compute the entropy value $\mathbb{K}$ as
\begin{equation} \label{AE4AK00abc}
\mathbb{K}=\hat{\mathsf{d}}+ \mathsf{E}
\end{equation}
where $\mathsf{E}$ is given in~\ref{additive}. Note that the term $\mathsf{E}$ reflects the point in Remark 1, stating that the key verification is possible only if the positions of discarded columns are known. Once these positions are known, inserting zeros within $\hat{\mathbf{d}}$ to extend the length would not require any additional information, hence the summation given in \ref{AE4AK00abc} is used instead of the actual key entropy. Note that the entropy of extended key (without accounting for entropy  needed to conduct verification) is equal to:
\begin{equation} \label{AE4AK00abcde}
\hat{\mathsf{d}}+ \log_2{\mathsf{\hat{\mathbf{d}}} \choose. \mathsf{p}}.
\end{equation}
Although \ref{AE4AK00abcde} is less than \ref{AE4AK00abc}, since \ref{AE4AK00abcde} does not include  the entropy required to  perform verification, expression \ref{AE4AK00abc} is used in deriving numerical results. In all cases studied here, the quantity in \ref{AE4AK00abc} is larger than the corresponding $\mathsf{SEC}$ level. 

\subsection{Properties of Error Vector}

To form the public key, Alice first forms $\mathbf{A}$ and $\mathbf{B}$ such that 
\begin{equation}
\mathbf{A}\mathbf{B}=
\begin{bmatrix}
\mathbf{I} & 0 \\
0  & \mathbf{E} 
\end{bmatrix}.
\label{Eq12AKK}
\end{equation}
Then, the public key is formed as (see Fig.~\ref{FNSet1})
\begin{equation}
\mathbf{P}=\mathbf{B}\mathbf{G}.
\label{Eq12}
\end{equation}
Bob removes $\mathsf{p}$ randomly selected columns from $\mathbf{P}$  to obtain the reduced public key matrix $\hat{\mathbf{P}}$. 
Key is encapsulated by Bob using (see Fig.~\ref{FNSet1})
\begin{equation}
\hat{\mathbf{P}}\hat{\mathbf{d}}+\mathbf{e}
\label{Eq12AKKH}
\end{equation}
where $\hat{\mathbf{d}}$ is of size $\mathsf{n-p}$, $\mathbf{e}$ is an error vector of size $\mathsf{n+m}$ formed as
\begin{equation}
\mathbf{e}=
\begin{bmatrix}
\mathbf{f} \\
0  
\end{bmatrix}
\label{Eq12AKKHA}
\end{equation}
and vector $\mathbf{f}$ is an error vector of size $\mathsf{n}$ formed based on the state diagram in Fig.~\ref{FigS}.   

Probability $\mu$ in Fig.~\ref{FigS} adjusts the entropy of the error vector 
$\mathsf{e}$, and the number of  ones (erroneous bits) in the error sequence. The larger is the number of ones, the harder will be the information set decoding. The higher is the entropy of $\mathsf{e}$, the exhaustive search for finding  $\mathsf{e}$  will be harder. This means there is a trade-off between number of errors in $\mathsf{e}$  and entropy of $\mathsf{e}$ . Parameter $\mu$ is selected to optimize this trade-off. Circled points in Figs.~\ref{Cross260},\ref{Cross311},\ref{Cross512},\ref{Cross1000} are selected based on this principle.

Note that in inserting error sequence $\mathsf{e}$,  ordering of repetition codes in columns is not publicly known, and the columns discarded by Bob are not known neither publicly, nor to Alice, but Alice and Bob will be able to arrive at the same key. This is achieved by Bob (privately) inserting zeros in locations corresponding to discarded columns, and at the Alice's side, the reconstructed key (upon multiplication by $\mathsf{A}$ and error correction) automatically inserts zeros in the same locations  (see Figs.~\ref{Pru1},\ref{Pru2}).  

As mentioned earlier, Alice extracts the key by multiplying the vector received from Bob by $\mathbf{A}$. This multiplication does not change the order of bits forming the sequence $\mathsf{e}$,  while it brings back the structure of $\mathsf{B}$ to its original form (excluding discarded columns), and accordingly, Alice can correct for errors caused by $\mathsf{e}$.  

\section{Comparisons with McEliece \& Niederreiter  Cryptosystems} \label{compAK}

Complexity aspects include: (i) Storage requirement for storing the public key and (ii) Computational complexity of key encapsulation and recovery (decoding of the underlying FEC). Our comparisons is different from typical security analysis presented in literature on cryptography. The reason is as follows.  
The main computational complexity in McEliece or Niederreiter 
cryptosystems concerns decoding of the underlying code, e.g., Goppa code. Over time, many elegant  techniques are developed, e.g., \cite{RN25} to \cite{RN73}, which aim at simplifying the attack on McEliece cryptosystem vs. a straightforward information set decoding given in~\cite{set}.  
There are also quite a few different attack methods for McEliece/Niederreiter cryptosystems with different trade-offs in terms of security level vs. attack complexity. In addition, there are many variants of McEliece/Niederreiter cryptosystems  based on different coding schemes, and even more variants for decoding in each case. These factors  make a straightforward comparison difficult. In other words, there is no general data available that would make a perfect comparison (accounting for decoding complexity as well as various matrix multiplications) possible.  {\em For these reasons, we rely on a metric for comparison that is to the disadvantage of the proposed method, as explained next.} In our case, decoding of repetition codes has a trivial complexity. For this reason, our comparison does not include the complexity of decoding of the underlying code. This omission results in significant underestimation of the complexity of McEliece or Niederreiter cryptosystems vs. that of the proposed method.  

To main consistency,  we rely on notations used in publications containing the results
reported in Tables~\ref{TabN1} and \ref{TabN2}. 
Table~\ref{TabN1} provides examples of size of the public key and security level provided by McEliece cryptosystems using Goppa code.  
Table~\ref{TabN2} provides examples of the complexity  and security level provided by McEliece cryptosystems using Goppa code.  
Examples for McEliece cryptosystems are mainly extracted from the proposals submitted to NIST (National Institute of Standards and Technology) \cite{Main-ref}.  
From reference \cite{Main-ref}, public keys in their variants of McEliece cryptosystem are summarized in Table~\ref{TabN1}. Entries in Table~\ref{TabN1} are significantly more complex (in terms of key size and decoding of the underlying error correcting code) as compared to the example of the proposed method in Table~\ref{TabN1AK}. Proposed method also generates longer key lengths.  
\begin{table}
\begin{center}
\begin{tabular}{|c|c|c|c|c|} \hline
Rows & $n$ & $k$  & $\mathsf{SEC}$ & Memory  \\ \hline \hline
1 & 1632 & 1269 & 80 &  0.46 Mbits \\
2 &2960 & 2288 & 128& 1.53 Mbits \\
\underline{3} & \underline{6624} & \underline{5129} & \underline{256} & \underline{7.67 Mbits} \\ \hline
\end{tabular} 
\caption{Examples for Goppa codes used in CCA2-secure variants of the McEliece
cryptosystem \cite{Main-ref}. Memory requirements for public keys are computed based on  a systematic generator matrices composed of $k(n-k)$ bits \cite{Main-ref}.  
$\mathsf{SEC}$ is computed merely based on information set decoding attack, while in the current article, $\mathsf{SEC}$ is the smaller of two quantities defined in expression \ref{AE4AK00}. } \label{TabN1}
\end{center}
\end{table}
Table~\ref{TabN1AK} provides a summary of the performance of the proposed method corresponding to Figs.~\ref{Cross260},\ref{Cross311},\ref{Cross512},\ref{Cross1000}. Comparing to Table~\ref{TabN1} with Table~\ref{TabN1AK}, it is observed that the proposed method produces significantly larger keys, while using much lower memory resources.  

\begin{table}
\begin{center}
\begin{tabular}{|c|c|c|c|} \hline
Size of   & Entropy    &  Size of  &   $\mathbb{K}$   \\
 ${\bf A,B}$  & of $\mathsf{SEC}$ & Public Key  &   from \ref{AE4AK00abc} \\
$\mathsf{n+m}$  & in bits  & $(\mathsf{m}+\mathsf{n})\mathsf{d}$  
 &  in bits  \\ \hline
625 &   \underline{{\bf 258}} & \underline{{\bf 0.128} MBits} & 320 \\	\hline
910 & {\bf 388} & {\bf 0.273} MBits & 467 \\	\hline
1210 & {\bf 524} & {\bf 0.484} MBits & 627\\	\hline
2260 & {\bf 1000} & {\bf 1.695} MBits & 1172 \\	\hline
\end{tabular} 
\caption{Performance and memory requirement of the proposed method where $\mathsf{n=3d}$. The $\mathsf{SEC}$ values defined in expression \ref{AE4AK00} correspond to circled points in 
Figs.~\ref{Cross260},\ref{Cross311},\ref{Cross512},\ref{Cross1000}.} \label{TabN1AK}
\end{center}
\end{table}

As mentioned, error detection in the proposed method has a trivial complexity, and as a result, the main complexity is that of the multiplication of an 
$\mathsf{(n+m)\times (d-p)}$ public key matrix by a data vector of size 
$\mathsf{d-p}$ (performed at the Bob's side for key encapsulation). Multiplication by matrix  $\mathsf{A}$  at the Alice side (for recovering the key) mainly involves an identity matrix forming the upper left corner of $\mathsf{A}$, which is small, and is not included here. The matrix multiplications at the Alice's side for McEliece  and Niederreiter Crypto-Systems are not included either, even though these would be more complex as compared to multiplication by matrix  $\mathsf{A}$ conducted as part of the key recovery in our proposed method. 

Comparisons of computational complexities presented next are based on multiplication of   the public key matrix by the message vector used in key encapsulation.
Relaying on notations used in the article, i.e., $\mathsf{n}$ and $\mathsf{k}$, for McEliece  and Niederreiter Crypto-systems, the corresponding complexity terms include multiplication of matrices of size $\mathsf{n\times k}$ and $\mathsf{n\times (n-k)}$ with message vectors of  size $\mathsf{k}$ and $\mathsf{n-k}$, respectively.  
These terms are computed as:
$\mathsf{n\times k^2}$ and $\mathsf{n\times (n-k)^2}$ bit operations, respectively. In our proposed method, a similar expression would result in a complexity of  $\mathsf{(n+m)\times (d-p)^2}$ bit operations.  Tables~\ref{TabN2AK} and \ref{TabN2} include the corresponding results. 
Underlined entries in row 3 of Tables~\ref{TabN1},\ref{TabN2} should be compared with underlined entries in Tables~\ref{TabN1AK},\ref{TabN2AK}, respectively. These comparisons are summarized in Table~\ref{TabNew1}. 

\begin{table}
\begin{center}
\begin{tabular}{|c|c|c|c|c|c|} \hline
$\mathsf{n}$ &  $\mathsf{m}$ & $\mathsf{d}$ & $\mathsf{p}$ & $\mathsf{SEC}$ & Complexity \\  \hline
620 & 10 &  205 & 80 & \underline{258} & \underline{$0.0783\times 10^{6}$}  \\	\hline
900 & 10 &  300 & 118 & 388 & $0.1656\times 10^{6}$ \\	\hline
1200 & 10 &  400 & 155 & 524 &  $0.2964 \times 10^{6}$ \\	\hline
2250 & 10 &  750 & 302 & 1000 & $1.0125\times 10^{6}$ \\	\hline
\end{tabular} 
\caption{Example for the complexity values for the the proposed method (limited to multiplication of the public key matrix by the message vector for key encapsulation, i.e., $\mathsf{(n+m)(d-p)}$ where $\mathsf{SEC}$ is defined in expression \ref{AE4AK00}.} \label{TabN2AK}
\end{center}
\end{table}

\begin{table} 
\begin{center}
\begin{tabular}{|c|c|c|c|} \hline
Rows & $\mathsf{SEC}$ & McEliece  & Niederreiter   \\  \hline
1 & 80 & $0.0215\times 10^{10}$  & $0.26\times 10^{10}$ \\ \hline
2 &128 & $0.1337\times 10^{10}$  & $1.55\times 10^{10}$ \\ \hline
3 & \underline{256} & \underline{$1.4805\times 10^{10}$}  & \underline{$17.43\times 10^{10}$}\\ \hline
\end{tabular} 
\caption{Example for the complexity values for the McEliece  and Niederreiter Crypto-systems (limited to multiplication by public key for key encapsulation). 
 Rows 1,2,3 correspond to rows 1,2,3 in in Table~\ref{TabN1}. $\mathsf{SEC}$ is computed based on information set decoding attack, while in the current article, $\mathsf{SEC}$ is defined in expression \ref{AE4AK00} (is the minimum of two terms, one of the two capturing information set decoding attack).} \label{TabN2}
\end{center}
\end{table}

\begin{table}
\begin{center}
\begin{tabular}{|c|c|c|c|} \hline
Methods & $\mathsf{SEC}$ & Memory  & Complexity \\  \hline
McEliece  & 256 & 7.67   & $1.4805\times 10^{10}$ \\   \hline 
Niederreiter & 256 & 7.67   & $17.43\times 10^{10}$ \\ \hline
Proposed & 258 & 0.128   & $0.0783\times 10^{6}$   \\ \hline
\end{tabular} 
\caption{Summary of comparisons. Memory is in MBits.} \label{TabNew1}
\end{center}
\end{table}

{\bf Remark 7}: Key recovery in McEliece  and Niederreiter  systems first results in a transformed key, which is then multiplied by the inverse matrix initially used for transformation to recover the original key (see expression \ref{APPr} in Appendix~\ref{mac}). Complexity of this transformation, although being significant, is not included in the above tables. Our proposed technique does not include such a transformation. $\blacksquare$

\newpage
\begin{center}
{\bf\LARGE Appendices}
\end{center}

\begin{appendix}


\section{McEliece Cryptosystem}\label{mac}

\subsection{Preliminaries}\label{mac}

McEliece Cryptosystem \cite{R0} is at the hearth of  randomized (quantum-safe) PKI techniques. The key idea is to deploy a linear code which, if rearranged in a particular form, enjoys a simple decoding method. Relying on communications between Alice and Bob, one of these two parties, say Alice, starts with a generator matrix for the underlying code that lends itself to a simple decoding algorithm. Then, Alice (privately) randomizes the generator matrix. The key idea is that Alice is able to revert the introduced randomness and then benefit from the known simple decoding algorithm.   Randomized generator matrix is known as the public key, and information that are (privately) available to Alice, which would enable removing the introduced randomness, is the private key. Alice sends the public key to Bob who encodes a random binary vector (secret message), introduces errors in the encoded message, and sends the result back to Alice. Then, Alice upon removing the randomness in the generator matrix, recovers the secret message. Classical McEliece Cryptosystem \cite{R0} relies on Goppa codes in which errors are added at the error correction capability of the code. 

In recent times, a number of other channel coding schemes, with their associated decoding methods,  have been applied to McEliece Cryptosystem~\cite{R13} to \cite{RN3}.   Some researchers have devised specialized decoding methods to enhance McEliece Cryptosystem~\cite{RN23} to \cite{RN65}.  A number of research work have examined hardware implementation of McEliece Cryptosystem \cite{RN16}\cite{RN4}, the use of codes other than Goppa code \cite{RN17} to \cite{RN74}, and various techniques to analyze and/or improve McEliece Cryptosystem~\cite{RN75} to \cite{RN80}. However, the memory requirement for the public key and complexity of key recovery remain to be challenging. 

In performing an exhaustive search, one needs a {\em Verification Apparatus} (VA) to test if the unknown key has been indeed discovered. In the following, this concept is explained in conjunction with the classical McEliece Cryptosystem and information set decoding which is the main attack technique for breaking McEliece Cryptosystem and its variations. 

In this Appendix, to simplify notations, discussions related to classical McEliece cryptosystem rely on italic bold fold notations, e.g.,  $\mathbi{A}$,   to represent matrices, $\eth$ represents the message vector, $\epsilon$ represents the added error vector, and notations such as $n$ and $m$ are used to show the underlying vector/matrix sizes.  In discussions related to the proposed method, regular boldface notations, e.g.,  $\mathbf{A}$  are used  to represent matrices, $\mathbf{d}$ represents the message  vector, $\mathbf{e}$ represents the added error vector, and notations such as 
$\mathsf{n}$ and $\mathsf{m}$ are used to show the underlying parameters.

Classical McEliece Cryptosystem is composed of the following components:
Public key is composed of the product of three matrices: $\mathbi{P}\mathbi{G}\mathbi{A}$ where 
$\mathbi{A}$ is a $k\times k$ invertible matrix;
$\mathbi{A}$ maps a $k$-dimensional data vector ${\eth}$ to another $k$-dimensional data vector $\breve{\eth}=\mathbi{A}{\eth}$.
Matrix $\mathbi{G}$ is the $n\times k$ generator matrix of an error correcting code (correcting up to $t$ errors), and $\mathbi{P}$ is an $n\times n$ permutation matrix. 

\subsection{Key Encapsulation and Recovery} 
Alice sends its public key, i.e., $\mathbi{P}\mathbi{G}\mathbi{A}$, to Bob,
Bob selects a random message vector ${\eth}$ and computes 
$\mathbi{W}=\mathbi{P}\mathbi{G}\mathbi{A}\eth+\epsilon$ where $\epsilon$  is an error vector containing $t$ ones and $n-t$ zeros. Bob sends $\mathbi{W}$ to Alice. Alice computes  $\mathbi{P}^{-1}\mathbi{W}=\mathbi{G}\mathbi{A}\,\eth+
{\breve\epsilon}=\mathbi{G}\,\breve{\eth}+{\breve\epsilon}$ where 
${\breve\epsilon}= \mathbi{P}^{-1}\epsilon$ is a permuted version 
of $\epsilon$, i.e., contains $t$ ones and $n-t$ zeros. Alice aims to recover the erroneous vector $\mathbi{G}\breve{\eth}+{\breve\epsilon}$ 
to obtain $\breve{\eth}$.
Alice then computes 
\begin{equation}
\eth=\mathbi{A}^{-1}\breve{\eth}
\label{APPr}
\end{equation}
to recover the encrypted data vector $\eth$. If $t\leq c$, where $c$ is the error correcting capability of the code generated by $\mathbi{G}$, then Alice succeeds in recovering $\eth$, i.e., the Encapsulated key.

\subsection{Information Set Decoding Attack} Noting 
$\mathbi{W}=\mathbi{P}\mathbi{G}\mathbi{A}\eth+\epsilon$, it is concluded that $\mathbi{W}$ and $\mathbi{P}\mathbi{G}\mathbi{A}$ are available to the attacker. Having access to these terms,  attacker aims to find $\eth$. The challenge facing the attacker is that some of the equations formed due to $\mathbi{P}\mathbi{G}\mathbi{A}\eth$ have an erroneous right hand side due to the addition of the error vector $\epsilon$ . In information set decoding, attacker exhaustively searches for a subset of equations of size $\mathbi{G}\mathbi{d}+\epsilon$ for which the components of error vector $\epsilon$ are zero. It then solves the corresponding equations. For every subset that is examined, the attacker needs to verify if its goal is achieved, i.e., {\em all} selected equations are error free. This is explained next.

\subsubsection{Verification Apparatus (VA) in Information Set Decoding}
Verification Apparatus can rely on one of two methods explained next.
\begin{enumerate}
 \vspace{-0.3cm}\item If data vector $\eth$ generates a code-word from code 
$\mathbi{P}\mathbi{G}$, it confirms that attack has ended in success. For each iteration of information set decoding, this is checked by: (a) computing 
$\mathbi{P}\mathbi{G}\mathbi{A}\eth$, (b) comparing the result with 
$\mathbi{W}=\mathbi{P}\mathbi{G}\mathbi{A}\eth+\epsilon$ to find $\epsilon$, 
(c) computing $\mathbi{P}\mathbi{G}\mathbi{A}\eth$ by removing $\epsilon$ from $\mathbi{W}$ and (d) computing the syndrome of $\mathbi{P}\mathbi{G}\mathbi{A}\eth$ by computing $\mathbi{s}=\mathbi{C}\mathbi{P}\mathbi{G}\mathbi{A}\eth$ where $\mathbi{C}$ is the parity check matrix associated with the generator matrix 
$\mathbi{P}\mathbi{G}\mathbi{A}$. Expressing the generator matrix and its associated parity check matrix in systematic forms, the computation of syndrome involves multiplication of an $(n-k)k$ matrix with a vector of size $n$ 
(this ignores the multiplication with the identity matrix included in $\mathbi{C}$). 
\vspace{-0.3cm}\item In a simpler alternative, $\mathbi{P}\mathbi{G}\mathbi{A}\eth$ is computed and compared with $\mathbi{W}=\mathbi{P}\mathbi{G}\mathbi{A}\eth+\epsilon$ to find $\epsilon$. Then, the number of bit errors in $\epsilon$ is counted. In most cases, one relies on forward error correcting codes with a fixed error correction capability (equal to the number of errors $t$). A condition for the attack to be successful is that, the number of bit errors, i.e., number of ones in $\epsilon$, should be equal to $t$.

\end{enumerate}

\section{Increasing Complexity of the Information Set Decoding} \label{APPP1}
In the current article, the forward error correcting code used in key encapsulation/recovery is composed of a concatenation of repetition codes of length 3. Alice randomly assigns the codes to columns of the generator matrix to construct matrix $\mathbf{G}_1$, and constructs  matrix $\mathbf{G}$ by appending random rows, $\mathbf{G}_2$, to $\mathbf{G}_1$ (see Fig.~\ref{FNSet2}). Public key is formed as  
$\mathbf{B}\mathbf{G}$ (see Fig.~\ref{FNSet2}). This is a matrix of size 
$\mathsf{(n+m)}\times \mathsf{d}$, which would be multiplied by a message vector 
$\mathbf{d}$ of size $\mathsf{d}$, and then an error vector $\mathbf{e}$ is added to the result. Such a construction follows all the steps used in the classical McEliece cryptosystem. 

Information set decoding works as follows: An attacker, having access to the public key, $\mathbf{B}\mathbf{G}$ , and vector $\mathbf{B}\mathbf{G}\mathbf{d}+\mathbf{e}$, exhaustively forms a system of linear equations by selecting a subset of of size $\mathsf{d}$ from the set of $\mathsf{(n+m)}$ equations.  Note that in the proposed method, unlike alternative schemes, $\mathbf{e}$ contains a random number of erroneous bits.
On the other hand, in the proposed method, columns of $\mathbf{G}$ each form a separate repetition code which are concatenated, resulting a single one in each row of 
the matrix $\mathbf{G}_1$. This feature allows Bob to select a random subset of columns from the public key $\mathbf{B}\mathbf{G}$ and discards them, reducing the column size and accordingly the size of the data vector to $\hat{\mathsf{d}}<\mathsf{d}$. In this case, the attacker will not be able to form the system of equations that is required in information set decoding attack, unless the positions of the discarded columns are first found through an exhaustive search. 

Using generic notations, consider  a chain of matrix products, say $\mathbi{QP}$, where $\mathbi{Q}$ is of size  $\upsilon\times \upsilon$ and  $\mathbi{P}$ is of size  $\upsilon\times \hat{\upsilon}$. 
Theorem \ref{ThNAK1} in Appendix \ref{ThNAK11} establishes how discarding of columns propagates in the chain. 

In classical McEliece Cryptosystem, the higher is the error correction capability of the code ${\bf G}$, i.e., $c$, the higher will be the number of equations in 
$\mathbi{P}\mathbi{G}\mathbi{A}\eth+\epsilon$ which can be turned erroneous by selecting a value of 1 for a component of $\epsilon$  in $\mathbi{W}=\mathbi{P}\mathbi{G}\mathbi{A}\eth+\epsilon$. This in turn adds to the complexity of information set decoding~\cite{set}. In this article, it is shown that by introducing memory among bits forming the error vector, the number of errors can be substantially increased, thereby adding to the complexity of information set decoding, while keeping the number of errors random. In particular, (i) by deploying a concatenation of repetition codes of length 3, and (ii) by relying on the state diagram in Fig.~\ref{FigS} for introducing memory among bits forming the error vector,  the number of errors can be increased, while the memory introduced by the state diagram results in detecting all inserted bit errors.   

In classical McEliece Cryptosystem, randomness is introduced by Alice who forms the public key. In the current article, Alice and Bob both introduce randomness. In the case of Alice, this is achieved by selecting the masking matrices and by permuting columns of the matrix generating the concatenation of repetition codes. In the case of Bob, randomness is introduced by discarded the columns of the public key and adding the error vector. 

Another difference is that, in classical McEliece Cryptosystem, to verify the success in a given round of information set decoding, the attacker can check the weight of the error event.  In the current article, attacker is not able to benefit from a first step verification by counting the number of errors (since it is random), necessitating a verification using the syndrome. This feature adds to the complexity of the attack. 

 In this work, a simple technique is applied such that the validity of a potential key cannot be verified by the decryption of a message (see Remark 1).

\section{Theorems}\label{Theo1}

To allow reusing the notations in different contexts, in this Appendix, we rely on generic notations  different from those used in the main body of the article.

\subsection{Discarding Columns in the Public Key Generator Matrix}\label{ThNAK11}

\begin{theorem} \label{ThNAK1}
Let us consider the product $\mathbi{\em QP}$, where matrix $\mathbi{\em Q}$ is of size $\alpha\times \beta$ and matrix $\mathbi{\em P}$  is of size 
$\beta\times \gamma$. The column indexed by $\nu$ is discarded from 
$\mathbi{\em P}$, i.e., columns indexed by $\nu+1$ to 
$\gamma$ replace columns 
$\nu$ to $\gamma-1$ and then column $\gamma$ is removed. Let us refer to the resulting matrix of size $\alpha\times (\gamma-1)$ as $\hat{\mathbi{\em P}}$.  A similar discarding and replacement  occurs in columns of $\mathbi{\em QP}$ for the same set of indices as in the case of $\mathbi{\em P}\rightarrow \hat{\mathbi{\em P}}$.
\end{theorem}
\begin{proof}
Let us use the notation $\mathbi{c}_\nu(\mathbi{P})$ and  $\mathbi{c}_\nu(\mathbi{PQ})$ to represent the column indexed by $\nu$ in matrices $\mathbi{P}$ and $\mathbi{QP}$, respectively. $\mathbi{c}_\nu(\mathbi{QP})$ is obtained by expanding $\mathbi{c}_\nu(\mathbi{P})$ over $\mathbi{Q}$, i.e., 
\begin{equation}
\mathbi{c}_\nu(\mathbi{QP})=\mathbi{Q}\times\mathbi{c}_\nu(\mathbi{P}).
\end{equation}
Consequently, removing $\mathbi{c}_\nu(\mathbi{P})$, results in removing $\mathbi{c}_\nu(\mathbi{QP})$ and shifting columns $\mathbi{c}_{\nu+1}(\mathbi{QP})$  to positions $\nu, ...\hat{\upsilon}-1$. 
\end{proof}

{\bf Remark 8:} Theorem \ref{ThNAK1} remains valid if the chain is formed by the product of multiple (more than two) matrices. For example, in a product of the form $\mathbi{S=RQP}$, if the column indexed by $\nu$ is removed from 
$\mathbi{P}$, then, in $\mathbi{S}$, columns indexed by $\nu+1$ to $\gamma$  replace columns $\nu$ to $\gamma-1$ and column $\gamma$ is removed. It is also easy to see that Theorem \label{ThNAK1} can be applied recursively to discard multiple columns.  
$\blacksquare$

{\bf Remark 9:} Matrix $\mathbi{P}$ in Theorem \ref{ThNAK1} represents the public key generator. Bob removes a randomly selected subset of columns from $\mathbi{P}$, and then multiply the  resulting reduced generator,  say  $\tilde{\mathbi{P}}$, with a message vector, which is accordingly shortened, before adding error. Then, Bob (in secret) extends the shortened data vector by inserting zeros in locations that were discarded (to produce an extended vector). At the Alice side, the vector received from Bob is multiplied by a second matrix, which results in three consecutive zeros in bit positions discarded from the data vector.  The vector formed at the Alice's side will be the same as a vector that would be obtained if, instead of discarding bits, Bob had selected a value of zero for them.  Consequently, Alice will be able to correct the error bits, and obtain Bob's extended data vector. 
 $\blacksquare$

\subsection{Sum of Binary Vectors}
\begin{theorem}\label{Th1}
Assume $\mathbf{x}$ and $\mathbf{y}$ are two independent binary vectors of length $k$ where $\mathbf{x}$ is unrestricted (is of entropy $k$). Regardless of $\mathbf{y}$ being restricted or not,  $\mathbf{z}=\mathbf{x}+\mathbf{y}$ will be unrestricted. 
\end{theorem}
\begin{proof}
\begin{eqnarray} \label{Ar1}
I(\mathbf{z};\mathbf{y}) & = & H(\mathbf{z})-H(\mathbf{z}|\mathbf{y}) \\
& = & H(\mathbf{z})-H(\mathbf{x}) \\
&=& H(\mathbf{z})-k.
\end{eqnarray}
Since $I(\mathbf{z};\mathbf{y})\geq 0$, we conclude $H(\mathbf{z})\geq k$ and since the entropy of a binary vector cannot be larger than its length, i.e.,  $H(\mathbf{z})\leq k$, we conclude $H(\mathbf{z})=k$, i.e., $\mathbf{z}$ is unrestricted. 
\end{proof}

\begin{theorem}\label{Th2}
Assume $\mathbf{X}$ and $\mathbf{Y}$ are two independent binary matrices of the same size, and  
\begin{equation} \label{EQR1}
\mathbf{Z}=\mathbf{X}+\mathbf{Y}.
\end{equation} 
We have
\begin{eqnarray} \label{Ar01}
H(\mathbf{X}|\mathbf{Z}) & = &  H(\mathbf{Y}) \\ \label{Ar02}
H(\mathbf{Y}|\mathbf{Z}) & = &  H(\mathbf{X}). 
\end{eqnarray}
\end{theorem}
\begin{proof}
Consider $\mathbf{Z}=\mathbf{X}+\mathbf{Y}$ and rewrite it as
$\mathbf{X}=\mathbf{Z}+\mathbf{Y}$.  
This means, for a given $\mathbf{Z}$, there is a one-to-one correspondence between 
$\mathbf{X}$ and $\mathbf{Y}$. In other words, for a given $\mathbf{Z}$, if $\mathbf{X}$ is changed to 
$\mathbf{X}_1=\mathbf{X}+\mathbf{A}\neq \mathbf{X}$, 
then, to realize the same $\mathbf{Z}$ in \ref{EQR1}, $\mathbf{Y}$ should be changed to  $\mathbf{Y}_1=\mathbf{Y}+\mathbf{A}\neq \mathbf{Y}$. As a result, for a given $\mathbf{X}$ and $\mathbf{Z}$, there is a unique $\mathbf{Y}$ satisfying~\ref{Ar01}. This establishes equality~\ref{Ar01}. Likewise, by exchanging roles of $\mathbf{X}$ and $\mathbf{Y}$, equality~\ref{Ar02} is established. 
\end{proof}

\subsection{Proof of Theorem~\ref{fullrank}} \label{ProFR}

{\bf Necessary Condition:} For simplicity of notations, let us assume $\mathsf{r}=3$. 
For $\mathsf{r}=3$, the expression in \ref{Eq2} reduces to:
\begin{equation}
\mathbf{Z}\mathbf{U}=\mathbf{c}_1\circledcirc \mathbf{r}_1+\mathbf{c}_2\circledcirc \mathbf{r}_2+\mathbf{c}_3\circledcirc \mathbf{r}_3.
\label{Eq2NNA}
\end{equation}
To have maximum entropy in expression~\ref{Eq2NNA}, any change in any of the three terms (atomic masks) should change the result of $\mathbf{Z}\mathbf{U}$. Consider a particular realization of matrix  $\mathbf{Z}$, denoted as 
${\cal R}_1^{\mathbf{Z}}$,  where 
\begin{eqnarray} \label{XY1}
\mathbf{c}_1 & = & \hat{\mathbf{c}}_1 \\ \label{XY2}
\mathbf{c}_2 & = & \hat{\mathbf{c}}_2 \\ \label{XY3}
\mathbf{c}_3 & = & \hat{\mathbf{c}}_3.  
\end{eqnarray}
Let us also assume there is a linear combination among columns of $\mathbf{Z}$ in realization ${\cal R}_1^{\mathbf{Z}}$ of the form
\begin{equation}
\hat{\mathbf{c}}_3=\alpha_1 \hat{\mathbf{c}}_1+\alpha_2 \hat{\mathbf{c}}_2.
\label{Eq2NNAK}
\end{equation}
Since matrix $\mathbf{Z}$ will be later restricted to have no all-zero column, and for a simple repetition of columns proof would follow similarly, we set 
$\alpha_1=\alpha_2=1$, i.e.,
 \begin{equation}
\hat{\mathbf{c}}_3=\hat{\mathbf{c}}_1+\hat{\mathbf{c}}_2.
\label{Eq2NNAK0}
\end{equation}
Let us fix $\mathbf{Z}$ at the realization ${\cal R}_1^{\mathbf{Z}}$; and sweep through possible values of $\mathbf{U}$, i.e., realizations of $\mathbf{r}_1$, $\mathbf{r}_2$ and $\mathbf{r}_3$. To maximize entropy, each realization of $\mathbf{U}$ should create a new mask. Let us  consider one such realization, ${\cal R}_1^{\mathbf{U}}$, as
\begin{eqnarray}   \label{XY4}
\mathbf{r}_1 & = & \hat{\mathbf{r}}_1 \\ \label{XY5}
\mathbf{r}_2 & = & \hat{\mathbf{r}}_2 \\ \label{XY6}
\mathbf{r}_3 & = & \hat{\mathbf{r}}_3
\end{eqnarray} 
where ${\cal R}_1^{\mathbf{U}}$ is full rank. 
For realization ${\cal R}_1^{\mathbf{U}}$, we have
\begin{equation}
\hat{\mathbf{c}}_1\circledcirc \hat{\mathbf{r}}_1 + 
\hat{\mathbf{c}}_2\circledcirc \hat{\mathbf{r}}_2 +
\hat{\mathbf{c}}_3\circledcirc \hat{\mathbf{r}}_3
\label{Eq2NNAK1}
\end{equation}
and
\begin{equation}
\hat{\mathbf{c}}_1\circledcirc \hat{\mathbf{r}}_1 + 
\hat{\mathbf{c}}_2\circledcirc \hat{\mathbf{r}}_2 +
(\hat{\mathbf{c}}_1+\hat{\mathbf{c}}_2)\circledcirc 
\hat{\mathbf{r}}_3,
\label{Eq2NNAK2}
\end{equation}
respectively.  Expression~\ref{Eq2NNAK2} reduces to:
\begin{equation}
\hat{\mathbf{c}}_1\circledcirc (\hat{\mathbf{r}}_1 + \hat{\mathbf{r}}_3)+
\hat{\mathbf{c}}_2\circledcirc (\hat{\mathbf{r}}_2 + \hat{\mathbf{r}}_3).
\label{Eq2NNAK2pp}
\end{equation}

Binary vectors produced by the term $\hat{\mathbf{r}}_1+\hat{\mathbf{r}}_3$ for different values of $\hat{\mathbf{r}}_1$ and $\hat{\mathbf{r}}_3$ 
form a group.  As a result, 
$\hat{\mathbf{r}}_1+\hat{\mathbf{r}}_3=\grave{\mathbf{r}}$ is a single binary vector of size $\mathsf{s}$ with at most $2^{\mathsf{s}}$ distinct  values, same is the case for 
$\hat{\mathbf{r}}_2+\hat{\mathbf{r}}_3=\acute{\mathbf{r}}$. Replacing in expression \ref{Eq2NNAK2pp}, we conclude 
 \begin{equation}
\hat{\mathbf{c}}_1\circledcirc (\hat{\mathbf{r}}_1 + \hat{\mathbf{r}}_3)+
\hat{\mathbf{c}}_2\circledcirc (\hat{\mathbf{r}}_2 + \hat{\mathbf{r}}_3)=
\hat{\mathbf{c}}_1\circledcirc \grave{\mathbf{r}}+\hat{\mathbf{c}}_2\circledcirc \acute{\mathbf{r}}.
\label{Eq2NNAK1pp}
\end{equation}
Now consider a second realization of  $\mathbf{U}$, denoted as 
${\cal R}_2^{\mathbf{U}}$, where
\begin{eqnarray}   \label{XYYY4}
\mathbf{r}_1 & = & \hat{\mathbf{r}}_1+\underline{\mathbf{r}} \\ \label{XYYY5}
\mathbf{r}_2 & = & \hat{\mathbf{r}}_2+\underline{\mathbf{r}} \\ \label{XYYY6}
\mathbf{r}_3 & = & \hat{\mathbf{r}}_3+\underline{\mathbf{r}}.
\end{eqnarray}
Expressions \ref{XYYY4}, \ref{XYYY5} and \ref{XYYY5} are for a given non-zero binary vector $\underline{\mathbf{r}}$ where $\underline{\mathbf{r}}\notin \{\hat{\mathbf{r}}_1,\hat{\mathbf{r}}_2,\hat{\mathbf{r}}_3\}$. It follows that
 \begin{equation}
\hat{\mathbf{r}}_1+\underline{\mathbf{r}}+\hat{\mathbf{r}}_3+\underline{\mathbf{r}}=\hat{\mathbf{r}}_1+\hat{\mathbf{r}}_3=\grave{\mathbf{r}}
\label{LLLAKK2}
 \end{equation}
 \begin{equation}
\hat{\mathbf{r}}_2+\underline{\mathbf{r}}+\hat{\mathbf{r}}_3+\underline{\mathbf{r}}=\hat{\mathbf{r}}_2+\hat{\mathbf{r}}_3=\acute{\mathbf{r}}.
\label{LLLAKKH2}
 \end{equation}
Replacing expressions \ref{XY1} to \ref{XY3} (capturing ${\cal R}_1^{\mathbf{Z}}$) and \ref{XYYY4} to \ref{LLLAKKH2} (capturing ${\cal R}_2^{\mathbf{U}}$)
in expression \ref{Eq2NNA} results in the same expression as in \ref{Eq2NNAK1pp}. Recall that expression \ref{Eq2NNAK1pp} was derived assuming realizations ${\cal R}_1^{\mathbf{Z}}$ and ${\cal R}_1^{\mathbf{U}}$. The final conclusion is that pair of realizations  $({\cal R}_1^{\mathbf{Z}}, {\cal R}_1^{\mathbf{U}})$ and 
 $({\cal R}_1^{\mathbf{Z}}, {\cal R}_2^{\mathbf{U}})$ have produced the same mask. 

Excluding realizations of $\mathbf{Z}$ which include all-zero columns(s), and likewise excluding realizations of $\mathbf{U}$ which include all-zero row(s), changes the conditions explained above\footnote{This means discarding realizations with linear combination(s) among columns of $\mathbf{Z}$ and/or among rows of $\mathbf{U}$.} to a stronger condition that acceptable realization of $\mathbf{Z}$ and $\mathbf{U}$ are limited to those with full rank, i.e., $\mathbf{r}$, in all acceptable realizations of $\mathbf{Z}$ and $\mathbf{U}$. 

\vspace{0.25cm}
\noindent 
{\bf Sufficient Condition:} Let us use the notations
${\bf \complement}_{\mathbf{Z}}$ to refer to the set of column-wise full-rank matrices of size $\mathsf{s}\times\mathsf{r}$, ${\bf \complement}_{\mathbf{U}}$ as the set of row-wise full-rank matrices of size $\mathsf{r}\times\mathsf{s}$,  and 
${\bf \complement}_{\mathbf{ZU}}$ as all elements formed when an element from 
${\bf \complement}_{\mathbf{Z}}$ is multiplied by an element from  
${\bf \complement}_{\mathbf{U}}$. Sufficient condition requires that the elements of  
${\bf \complement}_{\mathbf{ZU}}$  are unique.  Obviously, 
\begin{eqnarray}
\mathbf{Z}\in {\bf \complement}_{\mathbf{Z}} & \Longleftrightarrow & \mathbf{Z}\mathbf{L}\in {\bf \complement}_{\mathbf{Z}} \\
\mathbf{U}\in {\bf \complement}_{\mathbf{U}} & \Longleftrightarrow & \mathbf{L}\mathbf{U}\in {\bf \complement}_{\mathbf{U}}
\end{eqnarray}
where $\mathbf{L}\in {\bf \complement}_{\mathbf{L}}$ 
is an invertible matrix of size $\mathsf{r}\times \mathsf{r}$. 
Let us consider two realizations $(\hat{\mathbf{Z}}_1,\hat{\mathbf{U}}_1)$  and 
$(\hat{\mathbf{Z}}_2,\hat{\mathbf{U}}_2)$  of $(\mathbf{Z},\mathbf{U})$, where 
$\mathbf{Z}\in {\bf \complement}_{\mathbf{Z}}$ and 
$\mathbf{U}\in {\bf \complement}_{\mathbf{U}}$, such that
there is a unique $\acute{\mathbf{L}}\in {\bf \complement}_{\mathbf{L}}$ and a unique 
$\grave{\mathbf{L}}\in {\bf \complement}_{\mathbf{L}}$ satisfying 
\begin{eqnarray}
\hat{\mathbf{Z}}_2 & = &  \hat{\mathbf{Z}}_1\acute{\mathbf{L}}    \\
\hat{\mathbf{U}}_2 & = &  \grave{\mathbf{L}}\hat{\mathbf{U}}_1.
\end{eqnarray}
For ${\bf \complement}_{\mathbf{ZU}}$ to have repeated elements, e.g., 
$\hat{\mathbf{Z}}_1\hat{\mathbf{U}}_1=\hat{\mathbf{Z}}_2\hat{\mathbf{U}}_2$, we require
\begin{eqnarray}
\hat{\mathbf{Z}}_1\hat{\mathbf{U}}_1& =& \hat{\mathbf{Z}}_2\hat{\mathbf{U}}_2 \\
\hat{\mathbf{Z}}_1\acute{\mathbf{L}} \grave{\mathbf{L}}\hat{\mathbf{U}}_1& = & \hat{\mathbf{Z}}_2\hat{\mathbf{U}}_2\\
\hat{\mathbf{Z}}_1\tilde{\mathbf{L}}\hat{\mathbf{U}}_1& = & \hat{\mathbf{Z}}_2\hat{\mathbf{U}}_2
\end{eqnarray}
where $\tilde{\mathbf{L}}=\acute{\mathbf{L}} \grave{\mathbf{L}}$ is a unique element in ${\bf \complement}_{\mathbf{L}}$. Since $\tilde{\mathbf{L}}\in {\bf \complement}_{\mathbf{L}}$, we conclude $\hat{\mathbf{Z}}_2\hat{\mathbf{U}}_2$ is the same as an element obtained when $\mathbf{Z}$ is fixed at  $\hat{\mathbf{Z}}_1$, i.e., $\hat{\mathbf{Z}}_2=\hat{\mathbf{Z}}_1$, 
 and $\hat{\mathbf{U}}_2$ is set to $\tilde{\mathbf{L}}\hat{\mathbf{U}}_1$, or equivalently, when $\mathbf{U}$ is fixed $\hat{\mathbf{U}}_1$, i.e., 
$\hat{\mathbf{U}}_2=\hat{\mathbf{U}}_1$,  and $\hat{\mathbf{Z}}_2$ is set to 
 $\tilde{\mathbf{L}}\hat{\mathbf{Z}}_1$. 
These two equivalent cases produce a single mask matrix.
The only exception is the case that $\acute{\mathbf{L}}$ is a permutation matrix, and  
$\grave{\mathbf{L}}$ is equal  to its transpose (inverse). 
In this case, 
$\tilde{\mathbf{L}}=\acute{\mathbf{L}} \grave{\mathbf{L}}$ will be the identity matrix, even though the columns of $\mathbf{Z}_1$ and rows of $\mathbf{U}_1$ are changed (permuted using the same permutation). If $\acute{\mathbf{L}}$ and/or   
$\grave{\mathbf{L}}$ have more than a single one in a column and/or in a row, either 
$\hat{\mathbf{Z}}_1\acute{\mathbf{L}}$ and/or   
$\grave{\mathbf{L}}\hat{\mathbf{U}}_1$ would change in a non-trivial manner. This excludes the case of using a permutation matrix for $\acute{\mathbf{L}}$ and its inverse (transpose)  for $\grave{\mathbf{L}}$.

The collection of masks that relate to a product of the form 
$\hat{\mathbf{Z}}\hat{\mathbf{U}}$ by fixing $\hat{\mathbf{Z}}$ and changing 
$\hat{\mathbf{U}}$ to $\hat{\mathbf{L}}\hat{\mathbf{U}}$ (or vice versa) form a 
subset of the Ground Set, composed of unique elements, and the subset is closed when one scans through different 
$\hat{\mathbf{L}}\in {\bf \complement}_{\mathbf{L}}$. Now consider a 
$\bar{\mathbf{Z}}\bar{\mathbf{U}}$ that is not in the subset corresponding
to $\hat{\mathbf{Z}}\hat{\mathbf{U}}$. Such an element is unique and generates another subset of the Ground Set that does not have any common elements with the subset generated by  $\hat{\mathbf{Z}}\hat{\mathbf{U}}$.
This proves the uniqueness of elements in the Ground Set, when $\mathbf{Z}$ is column-wise full rank, and $\mathbf{U}$ is row-wise full rank.  
 $\square$

\subsection{Proof of Theorem~\ref{fullrank2}} \label{ProFR2}
Let us consider a masking matrix $\hat{\mathbf{M}}$ produced by a realization  
$\hat{\mathbf{Z}}\hat{\mathbf{U}}$ satisfying the conditions of Theorem~\ref{fullrank}, i.e, $\hat{\mathbf{M}}=\hat{\mathbf{Z}}\hat{\mathbf{U}}$ is of  rank $\mathsf{r}$. There are permutation matrices $\mathbf{P}_1$ and $\mathbf{P}_2$ which will relocate the matrix elements at the intersections (defined in Theorem~\ref{fullrank2}) to the upper left corner of  $\mathbf{P}_1\hat{\mathbf{M}}\mathbf{P}_2$.
The upper left corner of $\mathbf{P}_1\hat{\mathbf{M}}\mathbf{P}_2$ spans the space formed by columns, as well by rows, of $\hat{\mathbf{M}}$.
We have 
\begin{equation}
\mathbf{P}_1\hat{\mathbf{M}}\mathbf{P}_2=
\mathbf{P}_1{\mathbf{ZU}}\mathbf{P}_2
\end{equation}
where $\mathbf{P}_1{\mathbf{Z}}$ and ${\mathbf{U}}\mathbf{P}_2$ are another permissible realization of ${\mathbf{Z}}$ and ${\mathbf{U}}$, respectively. Discarding the trivial case that 
$\mathbf{P}_1$ and $\mathbf{P}_2$ are both identity matrices, we have 
$\mathbf{P}_1{\mathbf{Z}}\neq {\mathbf{Z}}$ and/or 
${\mathbf{U}}\mathbf{P}_2\neq {\mathbf{U}}$. This means, such a mask matrix occurs once in the Ground Set (assuming the trivial cases explained in Remark 4 are removed from the Ground Set). To verify this claim, let us focus on 
$\mathbf{P}_1{\mathbf{Z}}$ and ${\mathbf{U}}\mathbf{P}_2$. In $\mathbf{P}_1{\mathbf{Z}}$, columns corresponding to intersections in 
$\hat{\mathbf{M}}$ are moved to the upper $\mathsf{r}\times \mathsf{r}$ matrix in 
$\mathbf{P}_1{\mathbf{Z}}$.  In ${\mathbf{U}}\mathbf{P}_2$, rows corresponding to intersections in $\hat{\mathbf{M}}$ are moved to the right $\mathsf{r}\times \mathsf{r}$ matrix in ${\mathbf{U}}\mathbf{P}_2$. Let us write the matrices ${\mathbf{Z}}$ and 
${\mathbf{U}}$ as shown in Fig.~\ref{FigAAKK1}. Now let us consider all masks due to 
\begin{equation} \label{MMQW1}
\mathbf{P}_1{\mathbf{Z}}\mathbf{L}{\mathbf{U}}\mathbf{P}_2.
\end{equation} 
Scanning through all invertible $\mathsf{r}\times \mathsf{r}$ matrices $\mathbf{L}$ in
\ref{MMQW1} results in all possible realizations of the invertible $\mathsf{r}\times \mathsf{r}$ matrix $\check{\mathbf{Z}}\check{\mathbf{U}}$ in the upper left corner of $\mathbf{P}_1{\mathbf{ZU}}\mathbf{P}_2$ in Fig.~\ref{FigAAKK1}. Going through other  subsets of the Ground Set (due to other realizations of ${\mathbf{Z}}$ and/or 
${\mathbf{U}}$) that are not yet included, and for which the same intersection points form an $\mathsf{r}\times \mathsf{r}$ invertible matrix, result in repeating realizations of the matrix in the upper corner of the resulting mask matrix. Since each mask matrix in the Ground Set occurs once, given a mask ${\mathbf{M}}$ and 
non-trivial cases of $\mathbf{P}_1$ and $\mathbf{P}_2$, we have
\begin{equation} \label{MMQW2}
\mathbf{P}_1{\mathbf{M}}\mathbf{P}_2\neq {\mathbf{M}}.
\end{equation}
It follows that  each realization of the invertible $\mathsf{r}\times \mathsf{r}$ matrix 
in the upper left corner of the mask matrix $\mathbf{P}_1{\mathbf{ZU}}\mathbf{P}_2$ occurs with the same probability. $\square$

\begin{figure}[h]
   \centering
\hspace*{-0.8cm}
   \includegraphics[width=0.45\textwidth]{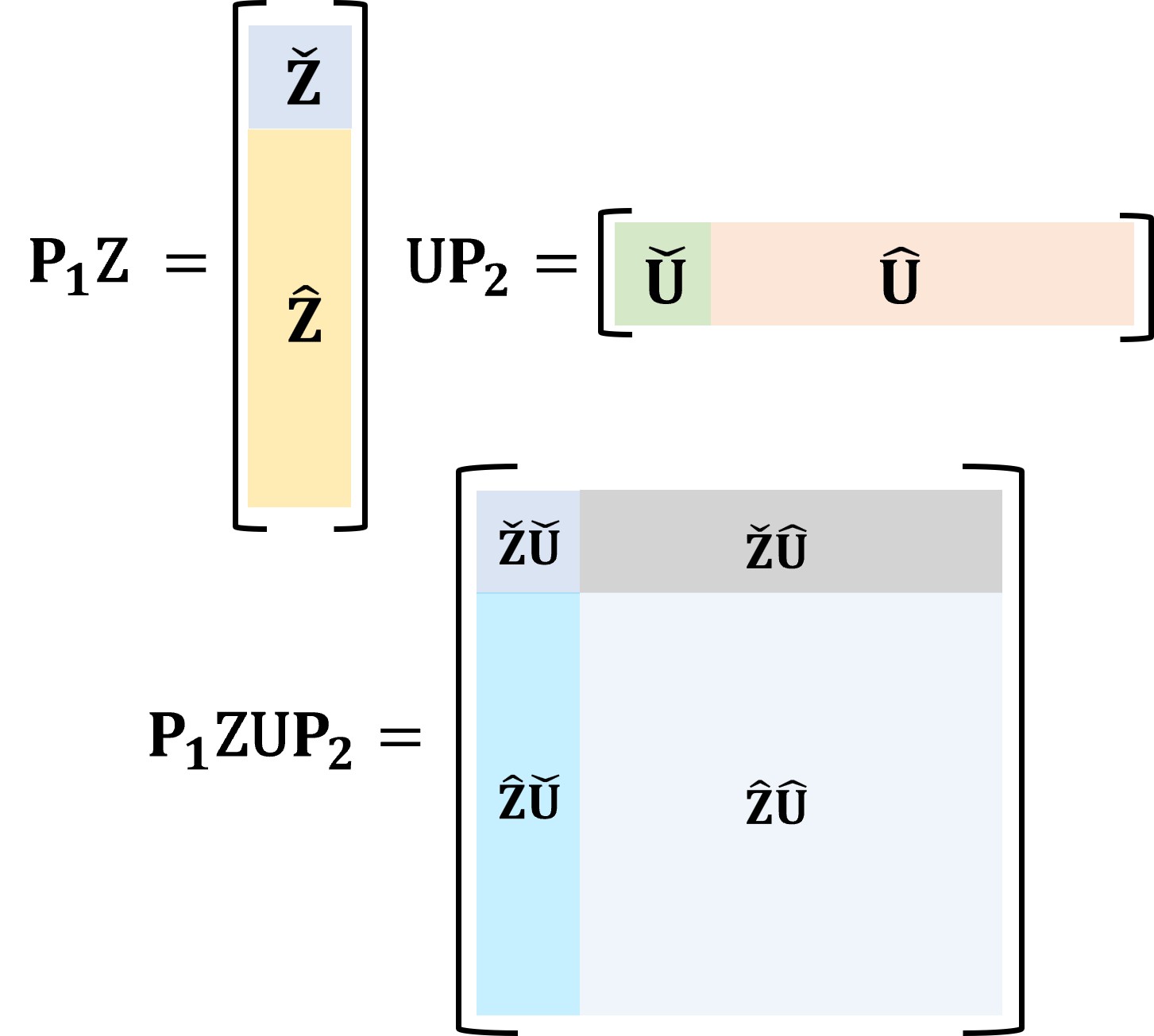}
   \caption{Matrix $\mathbf{P}_1{\mathbf{Z}}{\mathbf{U}}\mathbf{P}_2$.}
   \label{FigAAKK1}
 \end{figure}

\end{appendix}

\begin{small}

\end{small}

\end{document}